\documentclass[journal,onecolumn]{IEEEtran}
\usepackage{array}
\usepackage[caption=false,font=normalsize,labelfont=sf,textfont=sf]{subfig}
\usepackage{textcomp,setspace}
\usepackage{stfloats}
\usepackage{url}
\usepackage{verbatim,sarabian}
\usepackage{graphicx}
\usepackage{cite}
\usepackage{pifont,flushend,setspace}
\usepackage[misc,clock,geometry]{ifsym}
\usepackage{euscript}
\usepackage[switch]{lineno}

\usepackage[table]{xcolor}
\usepackage{tempora} 
\usepackage{cite,soul} 

\usepackage{mathtools,amscd}
\usepackage[normalem]{ulem}

\usepackage[inline, shortlabels]{enumitem}

\usepackage[scaled=0.85]{helvet}
\DeclareTextFontCommand{\ba}{\bfseries\sffamily}
\DeclareTextFontCommand{\hel}{\sffamily}

\setlist{itemjoin ={,\enspace},itemjoin* = {, and\enspace}}
\usepackage{lipsum}
\usepackage{amsthm}
\usepackage{amsopn}
\usepackage{amsmath,bm,amsfonts}
\usepackage{amssymb}
\usepackage{graphicx}
\usepackage{mathrsfs}
\usepackage{dsfont}

\usepackage{caption}
\usepackage{url,rotating,halloweenmath}
\usepackage{algorithm,algorithmic}

\usepackage[percent]{overpic}

\usepackage{tikz}
\usepackage{mathdots}
\usepackage{yhmath}
\usepackage{cancel}
\usepackage{siunitx}
\usepackage{array}
\usepackage{multirow}
\usepackage{gensymb}
\usepackage{tabularx}
\usepackage{extarrows}
\usepackage{booktabs}
\usetikzlibrary{fadings}
\usetikzlibrary{patterns}
\usetikzlibrary{shadows.blur}
\usetikzlibrary{shapes}

\usepackage[font=footnotesize]{caption}
\usepackage{hyperref}
\hypersetup{
colorlinks=true,
linkcolor=black,
filecolor=black,      
urlcolor=black,
citecolor = blue,
}

\newtheorem{theorem}{Theorem}

\definecolor{skyblue}{RGB}{48,116,178}
\definecolor{darkred}{RGB}{222,61,62}
\definecolor{darkgreen}{RGB}{67,160,88}
\definecolor{darkorange}{RGB}{223,145,51}
\definecolor{darkblue}{RGB}{51,91,171}
\definecolor{darkgreen}{RGB}{65,133,36}

\usepackage{stackengine}
\stackMath
\newlength\matfield
\newlength\tmplength
\def\matscale{1.}
\newcommand\dimbox[3]{
\setlength\matfield{\matscale\baselineskip}
\setbox0=\hbox{\vphantom{X}\smash{#3}}
\setlength{\tmplength}{#1\matfield-\ht0-\dp0}
\fboxrule=1pt\fboxsep=-\fboxrule\relax
\fbox{\makebox[#2\matfield]{\addstackgap[.5\tmplength]{\box0}}}
}
\newcommand\raiserows[2]{
\setlength\matfield{\matscale\baselineskip}
\raisebox{#1\matfield}{#2}
}
\newcommand\abmat[5]{
\stackunder{\dimbox{#1}{#2}{$\mathbf{#5}$}}{\scriptstyle#3\times #4}
}

\definecolor{crimson}{rgb}{0.86, 0.08, 0.24}
\definecolor{crimson}{rgb}{0.9, 0.13, 0.13}
\newcommand*{\textdoubletriangle}{
\resizebox{!}{8pt}{
\vbox{
\hbox{$\color{crimson}\blacktriangle$}
\nointerlineskip
\kern.5ex
\hbox{$\color{crimson}\blacktriangledown$}
}
}
}

\usepackage{inconsolata}

\renewcommand\hat\widehat
\renewcommand\geq\geqslant
\renewcommand\leq\leqslant

\newcommand{\npara}[1]		{\noindent {\bf #1}}
\newcommand{\bpara}[1]		{\smallskip \noindent {\bf #1}}

\usepackage{xspace}
\usepackage{dsfont}

\def\SE					            {-\textsarab{\SAhd}}

\def \fmax                     		{{f_{\mathsf{max}}}}
\def \fs                     			{{f_s}}

\def\adc					{ADC\xspace}
\def\snyq					{sub-Nyquist\xspace}
\def\sn					{\texttt{sNyq}\SE\xspace}
\def\snu					{\texttt{sNyq}$\lambda$\SE\xspace}
\def\rsnu					{$\rho$\texttt{sNyq}$\lambda$\SE\xspace}

\def\adcDR				    {DR\xspace}
\def\msef                   		{\EuScript{E}_2{\rob{\mat{f}_k,{\widetilde{\mat{f}}}_k}}}

\def\msem                   		{\EuScript{E}_\infty{\rob{\mat{f}_k,{\widetilde{\mat{f}}}_k}}}

\def\usf					{\texttt{USF}\xspace}
\def\khz					{\texttt{kHz}\xspace}
\def\hz					    {\texttt{Hz}\xspace}
\def\usfmc					{\texttt{USF-MC}\xspace}

\def\madcs				{$\mathscr{M}_\lambda$--\texttt{ADCs}\xspace}
\def\madc					{$\mathscr{M}_\lambda$--\texttt{ADC}\xspace}

\def\gn			                	{g[n]}
\def\gdn                			{g_{T_d}[n]}
\newcommand{\zm}[2]		{\xi_{#1}^{#2}}
\def\PO					{\boxed{\mathsf{P1}}}
\def\PT                			{\boxed{\mathsf{P2}}}
\newcommand{\vecd}[2]	 	{\Delta_{#1} \rob{#2}}
\def\sgn               			{\operatorname{sgn}}
\def\opt               			{\diamond}

\newcommand{\MO}[1]		{\mathscr{M}_\lambda\rob{#1}}
\newcommand{\MOl}[2]		{\mathscr{M}_{\lambda_{#2}}\rob{#1}}
\newcommand{\QO}[1]		{\mathscr{Q}_\lambda ({#1})}

\newcommand{\mat}[1]		{\mathbf{#1}}
\newcommand{\vect}[2]		{\mathrm{vec}_{#2}({#1})}

\newcommand{\rg}[1]		{\textcolor{black}{#1}}
\newcommand{\rgn}[1]		{\textcolor{black}{#1}}

\def\th	 {\mathscr{T}_{\sigma}}

\newcommand{\tr}[1] 	{\mathscr{S}^{\sigma}_{\lambda_{#1}}}
\newcommand{\id}[1]		{\mathbb{I}_{#1} }

\newcommand\fig[1]				{Fig.~\ref{#1}}
\newcommand\secref[1]			{Section \ref{#1}}
\newcommand\tabref[1]			{Table \ref{#1}}
\newcommand\algref[1]			{Algorithm \ref{#1}}

\newcommand\subfig[3]			{Fig.~\ref{#1} (${\mathsf{#2}}_{#3}$)}

\def\R					{\mathbb R}
\def\C					{\mathbb C}

\def\ind					{\mathds{1}}
\def\DE					{\stackrel{\rm{def}}{=}}
\def\etal					{\emph{et al}.~}
\def\eg					{\emph{e.g.~}}
\def\ie					{\emph{i.e.~}}
\def\viz					{\emph{viz.~}}

\newcommand{\QMs}		{\mathsf{H}}
\newcommand{\PMi}		{\mathsf{Q}_{i}}
\renewcommand{\H}		{\mathsf{H}}
\newcommand{\transp}{{\top}}

\def \dr               {DR\xspace}
\def \dk               {\mathscr{D}_K}

\def\iZ					{\in \mathbb Z}
\def\iR					{\in \mathbb R}

\def\yt					{y_{\lambda}}
\newcommand{\eqr}[1]		{\stackrel{\eqref{#1}}{=}}

\newcommand{\Lp}[1]{{\mathbf{L}}_{{#1}}}

\newcommand{\normt}[3]{ {\| {#1} \|}_{ {\Lp{#2}} \rob{#3}}}

\newcommand{\norm}[1]{\left\| {#1} \right\|}

\def\l						{\left(}
\def\r						{\right)}

\newcommand\rob[1]			{\l #1 \r}
\newcommand{\sqb}[1]		{\left[ #1 \right]}
\newcommand{\ft}[1]			{\left[\kern-0.15em\left[#1\right]\kern-0.15em\right]}
\newcommand{\fe}[1]		{\left[\kern-0.30em\left[#1\right]\kern-0.30em\right]}
\newcommand{\flr}[1]		{\left\lfloor #1 \right\rfloor}
\makeatletter
\newcommand*{\rom}[1]{\expandafter\@slowromancap\romannumeral #1@}
\makeatother
\usepackage[most,breakable]{tcolorbox}

\tcbset{
left = 0pt,
top = 0pt,
bottom = 0pt,
breakable,
before skip=0.2cm,
after ={\parindent0em},
colback=yellow!10!white, colframe=red!50!black, 
highlight math style= {enhanced, 
colframe=red,colback=red!10!white,
boxsep=0pt,
}}

\newtcbox{\abox}[1][brown]{on line,
arc=0pt,
colback=#1!10!white,
colframe=#1!50!black,
arc=0pt,
outer arc=0pt,
top=1pt,
bottom=0.5pt,
left=0mm,
right=0mm,
leftrule=0pt,
rightrule=0pt,
toprule=0.3mm,
bottomrule=0.3mm,
boxsep=0.1mm
}

\makeatletter
\def\moverlay{\mathpalette\mov@rlay}
\def\mov@rlay#1#2{\leavevmode\vtop{
\baselineskip\z@skip \lineskiplimit-\maxdimen
\ialign{\hfil$\m@th#1##$\hfil\cr#2\crcr}}}
\newcommand{\charfusion}[3][\mathord]{
#1{\ifx#1\mathop\vphantom{#2}\fi
\mathpalette\mov@rlay{#2\cr#3}
}
\ifx#1\mathop\expandafter\displaylimits\fi}
\makeatother

\renewcommand\bar\underline
\newcommand{\mse}[2]		{\EuScript{E}_2{\rob{\mat{#1},\mat{#2}}}}

\begin{document}

\title{
Sub-Nyquist USF Spectral Estimation:\\ 
$K$ Frequencies with $6K+4$ Modulo Samples
}

\author{Ruiming Guo, Yuliang Zhu and Ayush Bhandari

\thanks{This work is supported by the UK Research and Innovation council's \emph{Future Leaders Fellowship} program ``Sensing Beyond Barriers'' (MRC Fellowship award no.~MR/S034897/1).}

\thanks{The authors are with the Dept. of Electrical and Electronic Engineering, Imperial College London, South Kensington, London SW7 2AZ, UK. (Email: \texttt{\{ruiming.guo,yuliang.zhu,a.bhandari\}@imperial.ac.uk} or \texttt{ayush@alum.mit.edu}).}
}

\maketitle

{
\vspace{2pt}

\centering 
\fontsize{11}{11}\selectfont   
\color{blue} \ba{To appear, IEEE Transactions on Signal Processing (2024)}

\vspace{1cm}}

\begin{abstract}
Digital acquisition of high bandwidth signals is particularly challenging when Nyquist rate sampling is impractical. This has led to extensive research in sub-Nyquist sampling methods, primarily for spectral and sinusoidal frequency estimation. However, these methods struggle with high-dynamic-range (HDR) signals that can saturate analog-to-digital converters (ADCs). Addressing this, we introduce a novel sub-Nyquist spectral estimation method, driven by the Unlimited Sensing Framework (USF), utilizing a multi-channel system. The sub-Nyquist USF method aliases samples in both amplitude and frequency domains, rendering the inverse problem particularly challenging. Towards this goal, our exact recovery theorem establishes that $K$ sinusoids of arbitrary amplitudes and frequencies can be recovered from $6K + 4$ modulo samples, remarkably, independent of the sampling rate or folding threshold.
In the true spirit of sub-Nyquist sampling, via modulo ADC hardware experiments, we demonstrate successful spectrum estimation of HDR signals in the kHz range using Hz range sampling rates (0.078\% Nyquist rate). Our experiments also reveal up to a 33-fold improvement in frequency estimation accuracy using one less bit compared to conventional ADCs. These findings open new avenues in spectral estimation applications, e.g., radars, direction-of-arrival (DoA) estimation, and cognitive radio, showcasing the potential of USF.
\end{abstract}

\begin{IEEEkeywords}
Multi-channel architecture, robust recovery, spectral estimation, sub-Nyquist sampling, unlimited sampling.
\end{IEEEkeywords}

\tableofcontents

\newpage

\spacing{1.5}

\section{Introduction}
\label{sec:introduction}

\IEEEPARstart{S}{ub}-Nyquist Sampling \cite{Landau:1967:J,Feng:1996:C,Venkataramani:2000:J,Mishali:2010:J,Mishali:2011:J,Mishali:2011:Ja} is the umbrella term for scenarios where sampling at the Nyquist rate is impractical. This impracticality arises due to sensing bottlenecks including, 
\begin{enumerate}[leftmargin = *, label = $\bullet$]
\item high data rates caused by high-bandwidth signals, as seen in cognitive radios \cite{Quan:2008:J} and radar \cite{BarIlan:2014:J}, or,
\item when the hardware required for high-rate sampling is limited \cite{Hassanieh:2014:C} or prohibitively expensive \cite{Murmann:2015:J}.
\end{enumerate}

These issues have prompted new strategies utilizing inherent signal structure for lower-rate sampling. Foundational work in the area was conducted by Landau \cite{Landau:1967:J}, and later, by Bresler \& co-workers \cite{Feng:1996:C,Venkataramani:2000:J}. Following this, significant advancements have been made by Mishali \& Eldar\cite{Mishali:2010:J,Mishali:2011:J,Mishali:2011:Ja}, particularly in applications like wideband spectrum sensing, where high bandwidths are a significant bottleneck  \cite{Yen:2013:J,Fang:2021:J}. 

In the \snyq context, the sinusoidal model \cite{Stoica:2005:Book,Prony:1795:J, Cooley:1965:J, Robinson:1982:J, Xia:1999:J,Li:2009:J,Xiao:2018:J, Vaidyanathan:2011:J, Qin:2017:J},
\begin{align}
\label{eq:sos}
g\rob{t} = \sum\limits_{k = 0}^{K-1}c_k e^{\jmath \omega_{k} t}, 
\quad
\omega_{k} = 2 \pi f_k,
\end{align}
forms an important sub-class of the multi-band model and needs no introduction to the signal processing community. 
Frequency or Spectral Estimation is a significant research area due to its pervasiveness \cite{Stoica:2005:Book} with first solutions dating back to Prony's ingenious solution in 1795 \cite{Prony:1795:J}.
Numerous studies have focused on developing \snyq Frequency Estimation or \abox[brown]{{\sn}}\footnote{We use \sn to abbreviate \emph{\snyq Frequency Estimation} where the choice of \textsarab{\SAhd} symbolizes a tuning-fork that creates pure tones or frequencies.} methods based on multi-channel sampling. Some of these efforts on \sn predate the work on \snyq sampling. For instance, the pioneering works of Xia \etal \cite{Xia:1999:J,Li:2009:J,Xiao:2018:J} and Vaidyanathan, Pal \& co-workers \cite{Vaidyanathan:2011:J,Qin:2017:J}. 

\begin{figure}[!t]
\centering
\begin{overpic}[width=0.5\columnwidth]{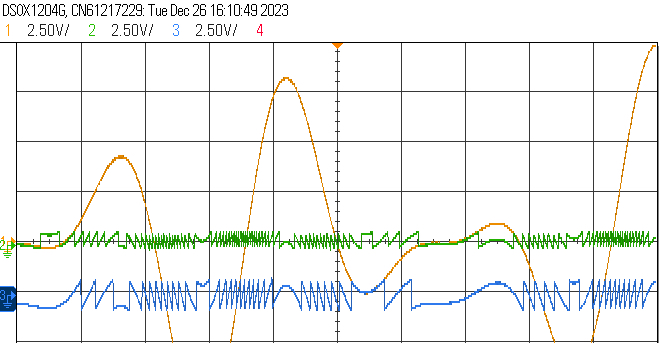}
\put (67,50) {\colorbox{white}{\sf{\bfseries\fontsize{6pt}{7pt}\selectfont {\color{darkorange}---- HDR Signal}}}}
\put (67,46) {\colorbox{white}{\sf{\bfseries\fontsize{6pt}{7pt}\selectfont {\color{darkblue}--- Channel-1 Modulo Signal}}}}
\put (67,42) {\colorbox{white}{\sf{\bfseries\fontsize{6pt}{7pt}\selectfont {\color{darkgreen}--- Channel-2 Modulo Signal}}}}
\end{overpic}
\caption{Oscilloscope view of our \usfmc hardware output (see \secref{sec:experiments}).}
\label{fig:screenshot}
\end{figure}

Despite the thrust of research spanning decades, almost all of the advances have been primarily focused on algorithmic novelty. 
The potential benefits of adapting the forward model or the acquisition pipeline, particularly in addressing practical challenges in HDR sensing \cite{Janzen:2017:J}, have remained unexplored.
To this end, our current research program pivoted on the Unlimited Sensing Framework (\usf) \cite{Bhandari:2017:C,Bhandari:2020:Ja,Bhandari:2021:J,Bhandari:2022:J}, and its applications in radar systems \cite{Feuillen:2023:C} and time-of-flight imaging \cite{Shtendel:2022:C}, provides a strong impetus for developing high-dynamic-range (HDR) sensing strategies for \sn.

\begin{figure}[!t]
\begin{center}
\includegraphics[width =0.65\textwidth]{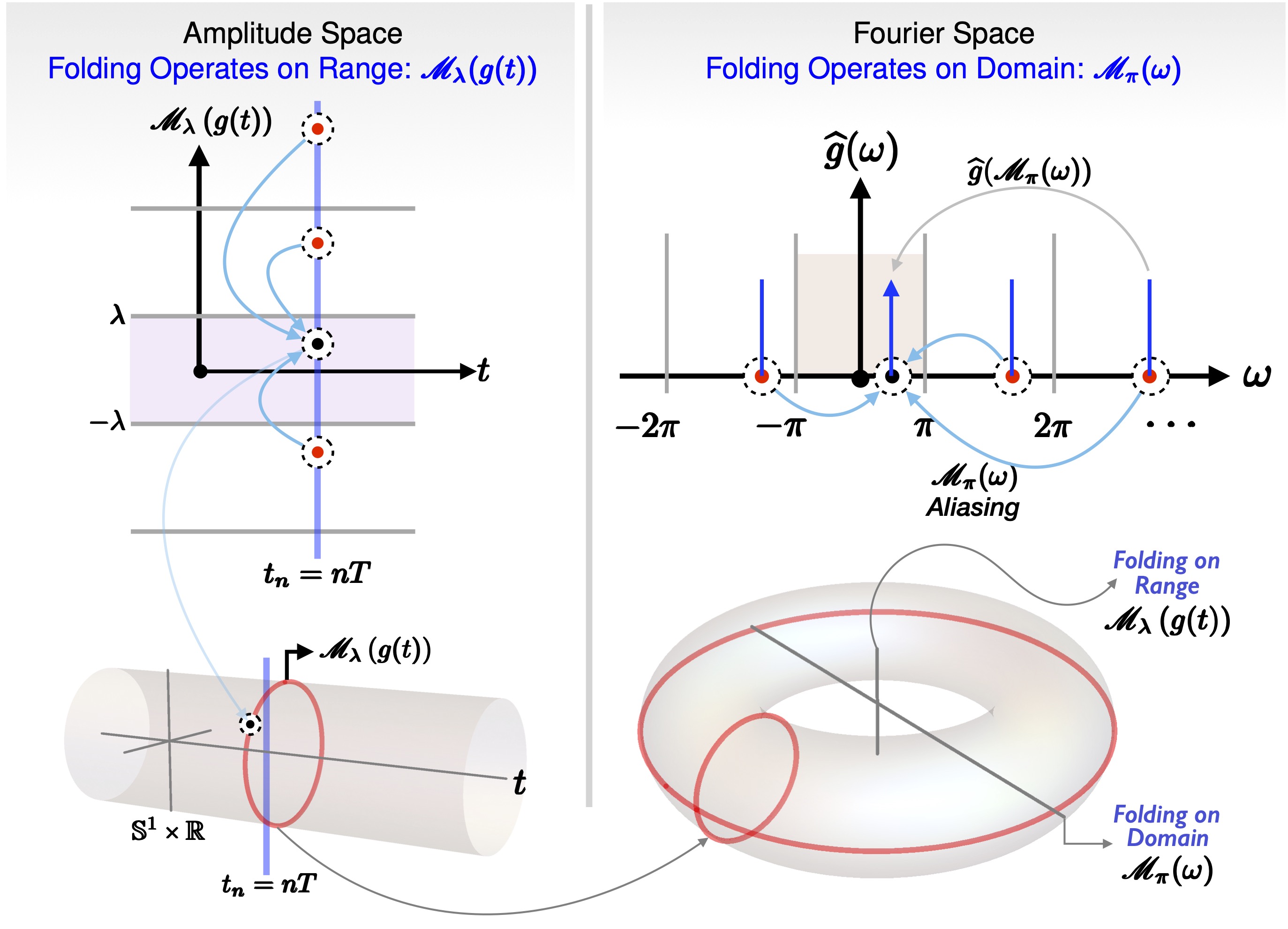}
\end{center}
\caption{
Range versus domain folding. In \usf, folding operates on the amplitude of the signal $\MO{g(t)}$ (see \fig{fig:screenshot}), yielding the quotient space. This is synonymous to how undersampling results in aliasing in the Fourier domain $\mathscr{M}_{\pi} \rob{\omega}$. Sub-Nyquist \usf spectral estimation (\snu) entails range and domain folding, simultaneously, creating a challenging inverse problem.}
\label{fig:MC}		
\end{figure}

\bpara{The Unlimited Sensing Framework.} The \usf is a new step in sampling theory \cite{Unser:2000:J}, fundamentally circumventing the implementation bottlenecks \cite{Murmann:2015:J} in the Shannon-Nyquist method.
The \usf embodies the hardware-software co-design principle, a cornerstone of computational sensing \cite{Bhandari:2022:Book}. Simultaneously achieving \emph{HDR acquisition} with \emph{high digital resolution}---a challenging trade-off in traditional paradigms \cite{Unser:2000:J,Murmann:2015:J}---makes the \usf particularly attractive.
This is achieved by folding an arbitrarily HDR signal in analog domain, as shown in \fig{fig:screenshot}, by deploying folding non-linearity in hardware. For this, the modulo \adc or \madc was proposed in \cite{Bhandari:2021:J}; $\lambda$ denotes the folding threshold. 
The low-dynamic-range, modulo \emph{encoded} samples are algorithmically \emph{decoded} to recover the HDR signal. As shown in \fig{fig:MC}, in the \usf, folding operates on the range of the input signal. This is akin to how undersampling results in aliasing in the Fourier domain.
The unique and wide benefits of the \usf have been explored in theory \cite{Bhandari:2018:C,Bhandari:2020:Ja,Bhandari:2021:J,FernandezMenduina:2021:J,Ordentlich:2018:J,Florescu:2022:J,Liu:2023:J,Shtendel:2023:J,Guo:2023:C,Guo:2023:Ca,Bhandari:2022:J,Beckmann:2022:J} and verified in practice \cite{Bhandari:2021:J,Beckmann:2022:J,Feuillen:2023:C,Beckmann:2024:J}. Hardware experiments with \madcs have shown that, 
\begin{enumerate}[leftmargin = 20pt, label = \ding{224}]
\item signals up to $25\lambda$ to $30\lambda$ \cite{Bhandari:2021:J} can be recovered in the presence of non-idealities, system noise and quantization.
\item a $10$-$12$ dB improvement in the quantization noise floor can be attained by simply replacing conventional \adc by \madc, in the context of radars \cite{Feuillen:2023:C} and tomography \cite{Beckmann:2024:J}. Application specific optimization of the \madc hardware is expected to perform better.
\end{enumerate}
These developments and the end-to-end implementation of the \usf substantiate its pervasive potential, partly attributed to the fundamental role played by sampling theory \cite{Unser:2000:J}.

\bpara{Related Works.}
Current works on the \usf are mainly focused on bandlimited (BL) signal spaces and rely on constant-factor oversampling \cite{Bhandari:2020:Ja,Bhandari:2021:J,Bhandari:2022:J}. 
Though not yet validated via hardware experiments, approaches based on prediction and side information can achieve near-Nyquist rates \cite{Romanov:2019:J}.
Beyond time-domain methods, a Fourier-domain method---\texttt{Fourier-Prony} algorithm---was proposed in \cite{Bhandari:2021:J} and validated via extensive experiments. This method leverages \emph{spectral separation} between BL signal and modulo folds in the Fourier domain. Variants of Fourier separation methods have been recently reported in \cite{Shah:2023:C} where frequency estimation in \texttt{Fourier-Prony} is replaced by $\ell_1$ minimization. Multi-channel \usf has been proposed in \cite{FernandezMenduina:2021:J,Gong:2021:J,Guo:2023:Ca}.

For sinusoidal mixtures \eqref{eq:sos}, assuming $\omega_{k} \in [-\pi,\pi]$, \usf based spectral estimation was considered in \cite{Bhandari:2018:C}. This method uses single channel and relies on substantial oversampling, both in the \emph{number of measurements} as well as the \emph{sampling rate}. Compared to Prony's method which uses $2K$ measurements, the approach in \cite{Bhandari:2018:C} uses $K' \geq 2K + \rob{7\sum_k|c_k|}/\lambda$ samples with sampling period $T\leq 1/2 \pi e$. Furthermore, this approach is based on higher order differences which may be unstable in practice. Can we do better? This motivates new methods that can enable {\bf \snyq \usf Spectral Estimation} or \abox{\snu}. Preliminary results outlining the \snu proof-of-concept, without the robust method and extensive hardware validation in this paper, \rgn{is} presented in \cite{Zhu:2024:C}. 

\bpara{Technical Challenges and Motivation.} While \snyq acquisition is our ultimate goal, we know that reconstruction methods at the core of \usf require oversampling\footnote{See \cite{Bhandari:2017:C,Bhandari:2020:Ja,Bhandari:2021:J} and the follow-up works \cite{Bhandari:2018:C,FernandezMenduina:2021:J,Beckmann:2022:J,Florescu:2022:J,Shtendel:2023:J,Ordentlich:2018:J,Liu:2023:J,Feuillen:2023:C,Beckmann:2024:J,Romanov:2019:J,Shah:2023:C}.}. Hence, \snyq acquisition results in a fundamental contradiction instigating a stalemate between \emph{analysis} and \emph{synthesis}. The key challenge in frequency estimation via \snyq \usf lies in the concurrent inversion of two kinds of non-linear folding operations. As shown in \fig{fig:MC},
\begin{enumerate}[leftmargin = 16pt, label = $\bullet$]
\item The first form of folding kicks in whenever $|g| >\lambda$. In this case, the folding operates on the range of the function. 
\item The second form of folding is induced due to undersampling. Even if the modulo folding is not triggered \ie $|g| \leq \lambda$, undersampling entails aliasing in the frequency domain. This is the case where folding operates on the domain\footnote{Prony's method enables estimation $K$ sinusoids in \eqref{eq:sos} by requiring only $2K$ samples to find the $2K$ unknowns, $\{c_k,\omega_k\}_{k=0}^{K-1}$ . However, Prony's method is susceptible to spectral aliasing. In the case of normalized frequency being $\pi$ (radians) or $T =1 $, it can only reliably estimate $\omega_k \in [-\pi,\pi]$. Frequencies beyond this range \ie $|\omega_k|>\pi$ will experience aliasing, folding back into the detectable range (see \fig{fig:MC}). The aliasing problem is solved via \sn methods \cite{Xia:1999:J,Li:2009:J,Vaidyanathan:2011:J,Qin:2017:J}.}.
\end{enumerate}
Owing to these opposing requirements, any sequential reconstruction, \ie \usf based unfolding \cite{Bhandari:2017:C,Bhandari:2020:Ja,Bhandari:2021:J} followed by \sn \cite{Xia:1999:J,Li:2009:J,Vaidyanathan:2011:J,Qin:2017:J} can not work. This highly challenging scenario motivates investigation of methods for \snu.

\bpara{Contributions.} The key takeaway from this work is the successful recovery of HDR signals in the \emph{kilohertz} bandwidth range through \snyq sampling at \emph{hertz} scale, utilizing the \madc hardware. Notably, our method shows recovery of signals up to $9\lambda$. The end-to-end implementation of the \snu approach requires that its core elements, \ie theory, algorithms, and hardware, work independently and in concert. To this end, our contributions are outlined as follows:

\setlength{\fboxsep}{2pt}
\setlength{\fboxrule}{0.5pt}

\begin{enumerate}[leftmargin=30pt,label = \fbox{\small$\textrm{C}_\arabic*$},itemsep = 2pt]
\item {\bf Theory} We conceptualize a quad-channel \usf architecture (\fig{fig:architecture}), the workhorse of our acquisition. Despite the folding non-linearity, akin to Prony's method, we prove a \snyq sampling theorem (see Theorem~\ref{thm:snUSF}) that is, (i) \uline{independent} of the sampling rate and (ii) enables recovery of $K$ sinusoids with \uline{only} $6K+4$ measurements. 

\item \textbf{Algorithms.} We design two novel algorithms. The first algorithm \abox{\snu} complements our sampling theorem. The second algorithm \abox{\rsnu}---a robust version of \abox{\snu}---is capable of handling hardware imperfections, quantization, and system noise in real-world scenarios.

\item \textbf{Hardware.} We develop a custom designed multi-channel \usf hardware (see \fig{fig:proto}) for validating our theory and algorithms. Extensive hardware experiments demonstrate its potential benefits across a wide range of scenarios.

\end{enumerate}

\begin{figure*}[!t]
\begin{center}
\includegraphics[width =0.98\textwidth]{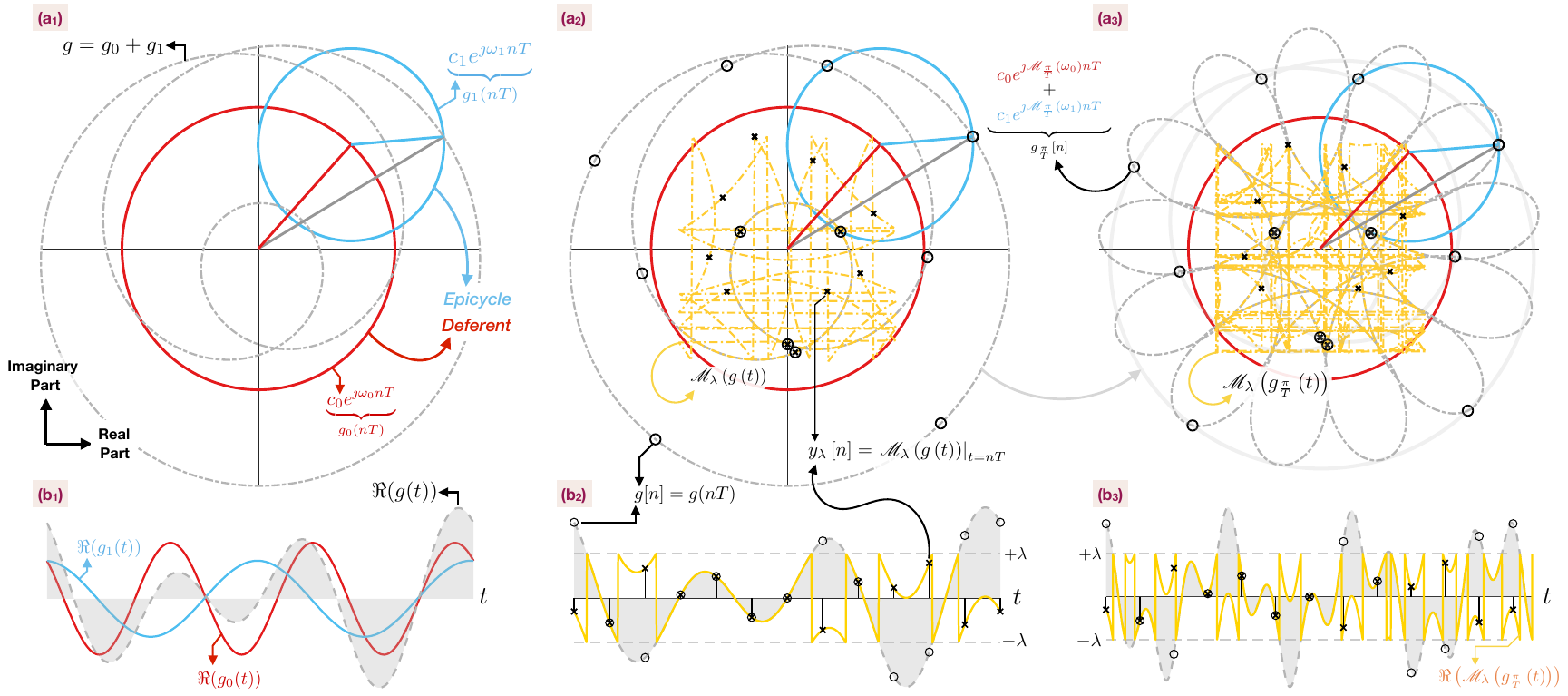}
\end{center}
\caption{Hipparchus--Ptolemy system of deferent-epicycle model with $g\rob{t} =  c_0 e^{\jmath \omega_0 t} + c_1 e^{\jmath \omega_1 t}$ ($K=2$). We also show, pointwise ($\circ$) and folded/\usf samples ($\times$).
(a) Polar-domain representation of $g\rob{t}$ (the Ptolemaic model). 
(b) Time-domain representation (real part of (a)). The goal is to estimate the frequency components $\omega_{0}, \omega_{1}$ from the \snyq \usf measurements. The challenges lie in the concurrent inversion of two forms of non-linear folding operations: 
(1) the frequency of the complex exponentials will be aliased once it is undersampled; 
(2) the amplitude of the signal will be folded whenever it exceeds the range $[-\lambda, +\lambda]$. This makes \snu ill-posed due to the inversion ambiguity as illustrated in ($\mathsf{a}_{2}$-$\mathsf{a}_{3}$) and ($\mathsf{b}_{2}$-$\mathsf{b}_{3}$) as well as in \fig{fig:MC}.}
\label{fig: phasor}
\end{figure*}

\bpara{Notation.} The set of integer, real, rational, irrational, and complex-valued numbers are denoted by $ \mathbb{Z}, \mathbb{R}, \mathbb{Q}, \mathbb{R} \backslash \mathbb{Q}$ and $\mathbb{C}$, respectively. The set of $N$ integers is given by $\id{N} = \{0,\cdots, N-1\}, N\in \mathbb{Z}^{+}$. 
The conjugate of $z\in \mathbb{C}$ is denoted by $\overline{z}$. 
The indicator function on domain $\mathcal{D}$ is denoted by $\ind_{{\mathcal{D}}}$. 
Continuous functions and discerete sequences are written as $g\rob{t}, t\in \mathbb{R}$ and $g\left[ n \right]$, respectively. 
We will be working with $T_d$--delayed sequences denoted by, 
${g_{{T_d}}}\left[ n \right] = {\left. {g\left( t \right)} \right|_{t = nT + {T_d}}}$. 
Vectors and matrices are written in bold lowercase and uppercase fonts, \eg $\mathbf{g} = [g[0],\cdots,g[N-1]]^{\transp} \iR^{N}$ and $\mat{G} = [g_{n,m}]_{n\in\id{N}}^{m\in\id{M}} \iR^{N\times M}$.
The $\Lp{p}\rob{\R}$ space equipped with the $p$-norm or $\normt{\cdot}{p}{\R}$ is the standard Lebesgue space; $\Lp{2}$ and $\Lp{\infty}$ denote the space of square-integrable and bounded functions, respectively. Spaces associated with sequences are denoted by $\ell^{p}$. 
The max-norm $(\Lp{\infty})$ of a function is defined as, $\norm{g}_{\infty} = \inf \{c_0 \geqslant 0: \left|g\rob{t}\right| \leqslant c_0 \}$; for sequences, we use, $\norm{g}_{\infty} = \max_{n} \left| g \sqb{n} \right|$.
The vector space of polynomials with degree less than or equal to $K$ is denoted by $P_{K}$. 
Function derivative is denoted by ${\partial _t}g\left( t \right)$ while for sequences,  the finite difference is denoted by $(\Delta g)\sqb{n} = g[n+1]- g\sqb{n}$. 
The Fourier Transform of $g\in \Lp{1}$ is defined by $\widehat{g} (\omega)  = \int g\rob{t} e^{-\jmath \omega t} dt$. The Discrete Fourier Transform (DFT) of a sequence $\mat{g}\in \ell_{1}$ is denoted by $\widehat{g}[m] = \sum\nolimits_{n=0}^{N-1} g\sqb{n} e^{-\jmath \frac{2\pi mn }{N}}$. Let $\mathbf{V}_{N}^{M}$ denote the $N\times M$ DFT/Vandermonde matrix $\mathbf{V}_{N}^{M}=\bigl[ \zm{N}{n\cdot m}  \bigr]_{n\in\id{N}}^{m\in\id{M}}, \ \zm{N}{n} = e^{-\jmath\frac{2\pi n}{N}}$. The DFT of $\mathbf{g}$ can be expressed algebraically as $\widehat{\mathbf{g}} = \mathbf{V}_{N}^{N} \mathbf{g}, \widehat{\mathbf{g}} \in \mathbb{C}^{N}$. The Kronecker product is denoted by $\otimes$. Element-wise smaller or equal is denoted by $\preccurlyeq$. Diagonal matrices are written as $\dk \rob{\mat{h}}$ with $\sqb{\dk \rob{\mat{h}}}_{k,k} = \sqb{\mat{h}}_{k\in\id{K}}$. 
We use $\rob{f\circ g}\rob{t}= f\rob{g\rob{t}}$ to denote function composition. The quantization operator is defined as $\QO{g}=2\lambda \flr{\rob{g +\lambda}/\rob{2\lambda}}$ where $\flr{g}  = \sup \left\{ {\left. {k \in \mathbb{Z}} \right|k \leqslant g} \right\}$ is the floor function. The centered modulo operator is defined as
\begin{equation}
\label{map}
\mathscr{M}_{\lambda}
:g \mapsto 2\lambda \left( {\fe{ {\frac{g}{{2\lambda }} + \frac{1}{2} } } - \frac{1}{2} } \right), 
\  \ft{g} \DE g - \flr{g}, \ \lambda>0.
\end{equation}
The mean squared error (MSE) between $\mat{x}, \mat{y} \iR^N$ is given by $\mse{x}{y} \DE \frac{1}{N}\sum\nolimits_{n=0}^{N-1} \left|x\sqb{n} -y\sqb{n} \right|^{2}$.

\section{Sub-Nyquist \usf Spectral Estimation}
\label{sec:methodology}

\npara{Problem Statement.} Let $g\rob{t} \in \Lp{\infty}$ denote a sum of complex exponentials signal as defined in \eqref{eq:sos},  
where $\{c_k, \omega_k\}_{k\in\id{K}}$ are the unknown amplitude and frequency, respectively. We assume $g\rob{t}\iR$ to align with the practice. The \usf eliminates clipping or saturation of potentially HDR signals by folding the analog input in hardware \cite{Bhandari:2021:J} via \eqref{map}. The resulting measurements in our case yield, 
$\yt\rob{t} =  {\mathscr{M}_{\lambda} \big({g\rob{t}}\big)}.$
Subsequently, the \emph{folded signal} $\yt\rob{t}$ is sampled in a pointwise fashion, 
\begin{equation}
\label{modulo samples}
\yt \sqb{n} = {\left.  \MO{g\rob{nT}} \right|_{t = nT}}, \quad n\in\id{N}
\end{equation}
where $f_s=\frac{\omega_s}{2\pi}=\frac{1}{T}$ is the sampling frequency. Our goal is to retrieve the parameters $\{c_k,\omega_k\}_{k\in \id{K}}$ from $\yt $. More importantly, we aim to develop the recovery strategy \emph{free from} $\omega_s$, truly paving the path for \snyq procedure.

\bpara{Hipparchus--Ptolemy Planetary Model: \\ Visual Illustration of Problem Statement.}	
We consider the \emph{deferent-epicycle model} \cite{Gallavotti:2001:J} of two planet system dating back to 2nd century AD. Let $g\rob{t} =  c_0 e^{\jmath \omega_0 t} + c_1 e^{\jmath \omega_1 t}$ where, in the \emph{Ptolemaic model}, $g_0\rob{t} = c_0 e^{\jmath \omega_0 t}$ and $g_1\rob{t} = c_1 e^{\jmath \omega_1 t}$ take the roles of deferent and epicycle, respectively. Our recovery problem is visually illustrated in the phasor/polar space in \fig{fig: phasor}. The equidistant samples and its $\MO{\cdot}$ folded versions are denoted by $\circ$ and $\times$, respectively. \subfig{fig: phasor}{b}{} shows the real part of \subfig{fig: phasor}{a}{}. The challenges in recovery arise from the concurrent inversion of folding along the \emph{range} and \emph{domain} (see \fig{fig:MC}), induced by \snyq-plus-\usf based sampling.
\begin{enumerate}[leftmargin = *]
\item Undersampling the analog signal results in frequency aliasing. This creates an ambiguity in frequency estimation as the measurements ($\times$) simultaneously correspond to multiple solutions as shown in \subfig{fig: phasor}{a}{2} and \subfig{fig: phasor}{a}{3}. 
\item The amplitude will be folded whenever it exceeds the \dr $[-\lambda,\lambda]$, which makes the inversion of $\MO{\cdot}$ ill-posed as multiple distinct signals could yield the same folded measurements as shown in \subfig{fig: phasor}{b}{2} and \subfig{fig: phasor}{b}{3}.
\end{enumerate}
The inversion of the folding operation $\MO{\cdot}$ is challenging inverse problem \cite{Bhandari:2017:C,Bhandari:2020:Ja,Bhandari:2021:J,Bhandari:2022:J}. Coupled with \snyq sampling, this turns into a highly ill-posed problem and motivates the design of a new methods that can simultaneously handle the two forms of folding non-linearities shown in \fig{fig:MC}.

\begin{figure}[!t]
\begin{center}
\includegraphics[width =0.65\textwidth]{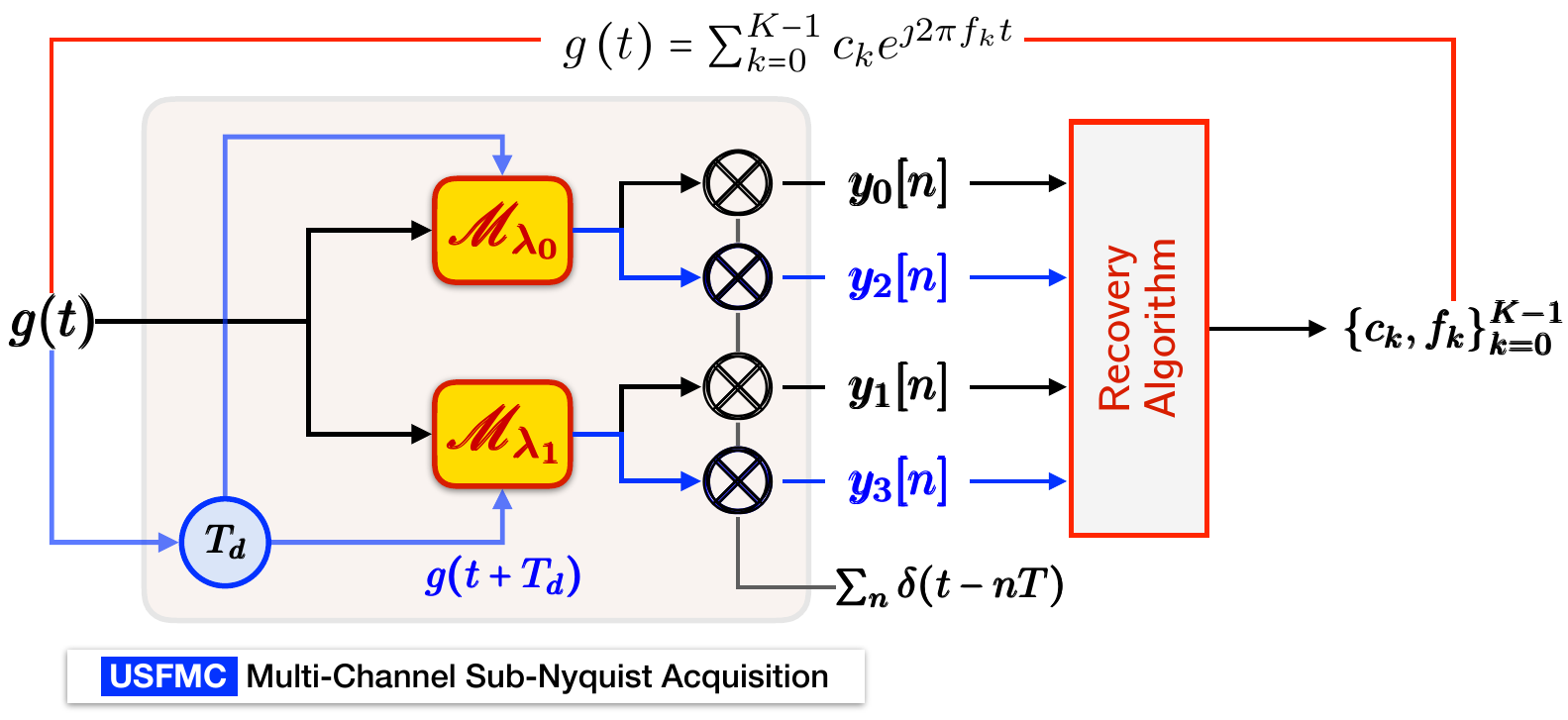}
\end{center}
\caption{The \usfmc pipeline for \snu requiring only $6K+4$ samples. It is independent of sampling rate and folding threshold. For hardware, see \fig{fig:proto}.}
\label{fig:architecture}
\end{figure}

\bpara{The \snu Method: Exact Recovery in the Ideal Case}

\renewcommand\descriptionlabel[1]{\hspace{-0pt}{\raisebox{-2pt}\FilledDiamondShadowC} \hspace{0.5pt}\bf{#1}}

\begin{description}[leftmargin = 0pt,itemsep = 4pt,labelwidth=!,labelsep = !]
\item{\abox{Mathematical Model of the Sensing Pipeline}} 

Our starting point is the mathematical modelling of the multi-channel or \usfmc pipeline in \fig{fig:architecture}. Denote by $\{y_{i}\}_{i\in \id{4}}$ the \usfmc measurements defined as, 
\begin{equation}
\label{eq:MCS}
\left\{ {\begin{array}{*{20}{l}}
y_{0} \sqb{n}= \MOl{\gn}{0} 
&
y_{1} \sqb{n} = \MOl{\gn}{1}
\\ 
y_{2} \sqb{n}= \MOl{\gdn}{0}
&
y_{3} \sqb{n} = \MOl{\gdn}{1}
\end{array}} \right.
\end{equation}
where 
$g\left[ n \right] = {\left. {g\left( t \right)} \right|_{t = nT}}$ and 
${g_{{T_d}}}\left[ n \right] = {\left. {g\left( t \right)} \right|_{t = nT + {T_d}}}$, and where $\{\lambda_0,\lambda_1,T_d\}$ are \usfmc design parameters.

\item{\abox{Simplification of the Measurements}}

Let $\bar{g} \DE \rob{\Delta g}$. For notational simplification, we define 
\begin{equation}
\label{eq:notation}
\left\{ {\begin{array}{*{20}{l}}
u_{0} = \bar{g}  - v_{0}, &v_{0} = {\Delta {y}_{\lambda_0}}   \\ 
u_{1} = \bar{g}  - v_{1}, &v_{1} = {\Delta {y}_{\lambda_1}} \\ 
u_{2} = \bar{g}_{T_d}  - v_{2}, &v_{2} ={\Delta {y}_{\lambda_0,T_d}}    \\ 
u_{3} = \bar{g}_{T_d}  - v_{3}, &v_{3} = {\Delta {y}_{\lambda_1, T_d}}  
\end{array}} \right.
\mbox{with} 
\begin{array}{*{20}{c}}
& \bar{g}_{i} = u_{i} + v_{i}, \\ 
& i\in \id{4}.
\end{array}
\end{equation}
To uncover the underlying mathematical structure of the \usfmc samples, we apply the modular decomposition property \cite{Bhandari:2021:J} in each channel which allows us to write,  
\begin{equation}
\label{eq:difference model}
u_{i} \sqb{n} =  \bar{g}_{i} \sqb{n} - v_{i}\sqb{n} \equiv   {\sum\limits_{l=0}^{L_i-1} {\gamma_{l,i} \delta [n - n_{l,i}]}}
\end{equation}
where $\{\gamma_{l,i},n_{l,i}\}_{l=0}^{L_{i}-1}$ are the unknowns that parametrize the residue signal, $u_i$.  Note that $\gamma_{l,i}\in 2\lambda_{i} \mathbb{Z}$, $\lambda_i = \lambda_{i\bmod 2}, i\in\id{4}$.

\item{\abox{Overview of the \snu Recovery Strategy}}

Since each channel $\{v_i\}_{i\in\id{4}}$ is undersampled, the folding non-linearity along amplitude cannot be inverted by non-linear filtering of amplitudes \cite{Bhandari:2020:Ja} or by Fourier-domain partitioning developed in \cite{Bhandari:2021:J}. 
Nonetheless, common to the theme of \usf is that idea of \emph{residue recovery} \cite{Bhandari:2020:Ja,Bhandari:2021:J} \ie given $v_i$, estimate $u_i$ so that $v_i + u_i \mapsto g_i$.
We use a similar idea which follows a different approach. The key insight being that $g_i$ is \emph{pairwise common to all channels}, its elimination in $v_0$ and $v_1$ results in recovery of the residue. Note that,  
\begin{equation}
\label{eq:residue difference}
v_{1}\sqb{n} - v_{0} \sqb{n} \eqr{eq:notation}
u_{0}\sqb{n} - u_{1} \sqb{n}   \equiv  \sum\limits_{l=0}^{L_0+L_1 - 1} {d_{l} \delta [n - n_{l}]}
\end{equation}
where $d_{l} =2\lambda_{1}\gamma_{l,1} - 2\lambda_{0}\gamma_{l,0}$ and the unknowns $\{\gamma_{l,0}, \gamma_{l,1}\} \in \mathbb{Z}^{2}$ can be uniquely estimated due to the assumption that ${\lambda_{0}}/{\lambda_{1}} \in \mathbb{R} \backslash \mathbb{Q}$. In effect, $v_{i} = \rob{ \Delta {y}_{\lambda_i}}, i=0,1$ provides for $\widetilde{u}_i$ in \eqref{eq:difference model} via \eqref{eq:residue difference} due to the co-irrationality of $\lambda_i$. Once $\widetilde{u}_i$'s are estimated, we can recover $\widetilde{g}_i = v_i + \widetilde{u}_i $.

Having obtained $\{\widetilde{g}_i \sqb{n}\}_{n\in\id{N-1}}$ for $i = 0,1$, provided that $N-1\geq 2K$, one can use Prony's method \cite{Prony:1795:J} to estimate the $2K$ unknowns in \eqref{eq:sos}. Nonetheless, the estimates of $\omega_k$ can still be aliased due to the \snyq sampling. This is where the additional channels, $i = 2,3$ play a vital role because we can use channel redundancy to recover the true (non-aliased) frequencies. To this end,
$v_2$ and $v_3$ comprising of $T_d$-delayed samples yield $\widetilde{g}_{T_d}\sqb{n}$, again following the residue recovery method above. There on, we use the $\{ {\widetilde {g} \sqb{n}},\widetilde{g}_{T_d}\sqb{n}\}$ to estimate non-aliased frequencies as in previous approaches \cite{Fu:2018:J}.

\end{description}

Our main result is summarized as follows.
\begin{theorem}
\label{thm:snUSF}
Let $g\rob{t} = \sum\nolimits_{k = 0}^{K-1}c_k e^{\jmath \omega_k t}$. Given multi-channel modulo measurements $\{y_i\sqb{n}\}_{n\in\id{N_{i}}}^{i\in\id{4}}$ defined in \eqref{eq:MCS}. Then, $g\rob{t}$ can be exactly recovered with $N_{i}\geqslant \rob{2-\flr{\frac{i}{2}}}K+1, i\in\id{4}$ samples if $\tfrac{\lambda_0}{\lambda_1} \in \mathbb{R} \backslash \mathbb{Q} $ and $ T_d \leq \frac{\pi}{\rob{\max_k |\omega_k|}}$.
\end{theorem}
\begin{proof}
Our proof is constructive and is based on the inversion of the \emph{range} and \emph{domain folding} operators. 

\bpara{i) Range Unfolding.} Given modulo samples $\{v_{0},v_{1}\}_{n\in\id{N_{i}}}$, $\tfrac{\lambda_0}{\lambda_1 }\in \mathbb{R} \backslash \mathbb{Q}$, we obtain $\{u_{0},u_{1}\}_{n\in\id{N_{i}}}$ via $$\rob{v_{1}- v_{0}} \sqb{n} \eqr{eq:residue difference}  \sum\limits_{l=0}^{L_0+L_1 - 1} {d_{l} \delta [n - n_{l}]},$$ leading to a unique decomposition. This can be proved by contradiction; suppose $\exists\{\alpha_{l,0}, \alpha_{l,1}\} \iZ$ such that $d_{l} =2\lambda_{1}\alpha_{l,1} - 2\lambda_{0}\alpha_{l,0}$, which results in $\tfrac{\gamma_{l,1}-\alpha_{l,1}} {\gamma_{l,0}-\alpha_{l,0}}  =   \frac{\lambda_{0}}{\lambda_{1}}.$ This contradicts the assumption $\tfrac{\lambda_0}{\lambda_1 }\in \mathbb{R} \backslash \mathbb{Q}$ as $\frac{\gamma_{l,1}-\alpha_{l,1}} {\gamma_{l,0}-\alpha_{l,0}} \in \mathbb{Q}$. Finding $\{\gamma_{l,i}\}$ from $d_l$ can be implemented via dictionary or brutal force search. Hence, $\{{u}_0, {u}_1\}$ can be recovered using \eqref{eq:difference model}; the same applies to $\{ {u}_2, {u}_3 \}$ via $\{{v}_2, {v}_3\}$. 

\bpara{ii) Domain Unfolding.} $\{g_{i}[n]\}_{n\in\id{N_{i}}}^{i\in\id{4}}$ can be expressed as 
\begin{equation}
\label{eq:delay}
\bar{g}_{i}\sqb{n} = \sum\limits_{k \in\id{K}}{c}_{k,i} e^{\jmath{\nu_k}T n},  
\begin{array}{*{10}{l}}
{c}_{k,i} = c_k(e^{\jmath {\omega_k}T }-1)e^{\jmath\rob{  {\omega_k}T_d \flr{\frac{i}{2}}}} \\ 
\nu_k = \mathscr{M}_{\omega_s} \rob{\omega_k}
\end{array}.
\end{equation}
The common frequencies $\{ \nu_k \}_{k\in\id{K}}$ can be found by using Prony's method \cite{Prony:1795:J} as follows. Let $h \sqb{n}$ be the filter with $z$-transform
$
\H (z)  = \sum_{n=0}^{K} h \sqb{n} z^{-n} = \prod_{k=0}^{K-1} (1 - u_k z^{-1})$ and roots $u_k = e^{\jmath { {\nu_k}}T} $. Then, $h \sqb{n}$ annihilates $\bar{g}_{i}\sqb{n}$ or, 
\[(h \ast \bar{g}_{i})\sqb{n} = \sum\limits_{m=0}^{K} h [m]  \bar{g}_{i}[n-m] 
= \sum\limits_{k=0}^{K-1} {c}_{k,i}  
\rob{\sum\limits_{m=0}^{K} h [m] u_k^{-m}}
u_k^{n}  = 0\] 
because $\mathsf{H}\rob{u_k} = 0$. The annihilation filter coefficients can be found by solving a system of linear equations. The aliased frequencies $\nu_k$ can be obtained by computing the zeros of the polynomial $\H \in P_{K}$. The non-aliased frequencies can be determined by computing the phase difference $e^{\jmath {\omega_k}T_d}$ in~\eqref{eq:delay}, provided that 
$\rob{\max_k |\omega_k|} T_d \leqslant \pi \Longleftrightarrow   T_d \leqslant \frac{\pi}{\rob{\max_k |\omega_k|}}$. The problem can be solved once there are at least many equations as unknowns; \ie $N_{0}, N_{1} \geqslant 2K+1$ samples for estimating $2K$ unknowns $\{\nu_{k}, {c}_{k,i}\}_{k\in \id{K}}$ and $N_{2}, N_{3} \geqslant  K+1$ samples for estimating $K$ unknown phases $\{e^{\jmath  {\omega_k}T_d}\}_{k\in \id{K}}$. 
This completes the proof. 
\end{proof}

\section{Sub-Nyquist \usf Spectral Estimation in Practice }
\label{sec:robust reconstruction}

Theorem \ref{thm:snUSF} provides guarantees for \snu independent of the sampling rate. However, 
the fragility of the assumption $\tfrac{\lambda_0}{\lambda_1} \in \mathbb{R} \backslash \mathbb{Q} $ with the challenges of the real-world scenarios may compromise HDR capability of the \snu method. Challenges we have identified via hardware experiments include, 
\begin{enumerate}[leftmargin=*, label = $\bullet$] 
\item {\bf Noise.} 
System and quantization noise, inherent in electronic circuits, inevitably lead to a mix of bounded and unbounded noise. This combination can cause outliers in both residue separation \eqref{eq:residue difference} and spectral estimation.
\item {\bf Hardware Imperfections.} Non-ideal foldings may occur in hardware implementation of folding non-linearities in the analog domain, resulting in off-grid folds or jumps of $u_{i}, i\in \id{4}$; \ie $\gamma_{l,i} \notin  2\lambda_{i\bmod 2}\mathbb{Z}$. This may compromise the performance of the recovery approach.
\item {\bf Experimental Calibration.} The precision of threshold calibration is limited up to certain digits, where the irrational threshold settings cannot be attained. This affects the robustness of residue separation \eqref{eq:residue difference}. 
\end{enumerate}

This necessitates the development of robust recovery methods that can tackle noise and hardware imperfections, specially when operating at \snyq rates. To this end, next, we introduce \abox{\rsnu} which is a robust version of \snu.

\bpara{Towards \rsnu: Robust Recovery Algorithm.} Here, we leverage properties that were previously not utilized. These include,  
(i) $\bar{g}_{i}\in \mathsf{span} \{u_k = e^{\jmath { {\nu_k}}T} \}_{k\in \id{K}}, i\in\id{4}$,
(ii) $g_i = g_{i+1}$,$i=\{0,2\}$.
Combining these properties result in, 
\begin{align}
\label{eq:global description}
\bar{g}_{i}\sqb{n} & = u_{i}\sqb{n} + v_{i}\sqb{n} = \sum\limits_{k = 0}^{K-1}c_{k,i}u_k^{n}, \quad i\in \id{4} \\ 
c_{k,0} & = c_{k,1} =\breve{c}_{k} , \ \ \mbox{ and }  \ \
c_{k,2} = c_{k,3} = \breve{c}_{k} e^{\jmath {\omega_k}T_d} \quad
\breve{c}_{k}  = c_k\rob{e^{\jmath  {\omega_k}T }-1} \quad u_k =  e^{\jmath { \nu_k}T }.
\label{eq:identical amplitude} 
\end{align}
In summary, this allows us to develop an optimization method utilizing a \emph{global signal model} across all channels \eqref{eq:global description}, moving from channel-wise processing to a an approach that
jointly harnesses all the channels and capitalizes on the uncorrelated nature of inter-channel noise and folding non-idealities.

Combining all channels, a concise representation reads,
\begin{equation}
\label{eq:block matrix}
\mathbf{\bar{g}}  = \left( {\bf I}_{4}\otimes \mathbf{\Theta}\right) \mathbf{c} 
\end{equation}
where, 
$\mathbf{\Theta} = \bigl[ u_k^{n}  \bigr]_{n\in\id{N-1}}^{k\in\id{K}}$,
$\mathbf{c}_{i} = \left[ c_{0,i},\cdots, c_{K-1,i}\right]^{\transp}$ 
and,
$\mat{c} = \vect{\mat{c}_i}{I}$  and $\bar{\mat{g}} = \vect{\mat{\bar{\mat{g}}}_i}{I}$ and $\vect{\mat{x}}{K}$  denotes vectorization of vectors $\{\mat{x}_k\}_{k\in\id{K}}$. This summarizes our data model.

In presence of distortions such as  quantization resolution, system noise and hardware mis-match, we can only recover the signal up to a tolerance level of $\sigma$ (measurement distortion). In view of \eqref{eq:identical amplitude}, in the noiseless case, we see that $c_{k,0}  = c_{k,1}$ and $c_{k,2} = c_{k,3} $. Hence, in the noisy case, the 2 channels may only differ by a maximum deviation of $\sigma$ leading to the constraints, 
$\norm{ \mathbf{u}_{i}+\mathbf{v}_{i} -\mathbf{u}_{i+1}-\mathbf{v}_{i+1} }_{\infty}\leqslant \sigma$, $i= \{0,2\}$, which acts as a regularization term. Hence, in the real-world scenario, the joint recovery problem can be posed as:
\begin{tcolorbox}[ams align,breakable]
\label{eq:global problem}
&\mathop {\min}
_{\substack{
\mat{\Theta},
\mat{c},
\mat{u}_i}}
\     \norm{\mathbf{\bar{g}}  - \left( {\bf I}_{4}\otimes \mathbf{\Theta}\right) \mathbf{c}}_{2}^{2}, \ \mbox{s.t.} \ \mat{u}_{i} \in 2\lambda_{i} \mathbb{Z} \nonumber \\
&  \norm{ \mathbf{u}_{i}+\mathbf{v}_{i} -\mathbf{u}_{i+1}-\mathbf{v}_{i+1} }_{\infty}\leqslant \sigma,  \ i= \{0,2\}.
\end{tcolorbox}
The minimization in \eqref{eq:global problem} is non-trivial, due to the structure and the constraints of the setup.  
In order to tackle this problem, we opt for an alternating minimization strategy where the goal is to split \eqref{eq:global problem} into two tractable sub-problems, \viz $\PO$ that addresses recovery of $\mat{\Theta}, \mat{c}$ via joint spectral estimation and $\PT$ that solves for $\mathbf{u}_{i},i\in\id{4}$ via robust residue separation. 

\subsection{Sub-Problem: Joint Spectral Estimation.} 
\label{sec:spone}
\bpara{Sparse Representation Model.} Assuming that $\mathbf{\bar{u}}_{i}, i\in\id{4}$ is known, it remains to estimate $\{\mat{\Theta}, \mat{c}\}$ by minimizing,
\begin{align}
\label{eq:vector frequency estimation}
\PO \quad \quad  \mathop {\min}
_{\substack{
\mathbf{\Theta},	\mathbf{c}}} \quad \  &   \norm{\mathbf{\bar{g}}  - \left( {\bf I}_{4}\otimes \mathbf{\Theta}\right) \mathbf{c}}_{2}^{2}
\end{align}
which is a non-convex problem. To find the solution, we utilize the following parameterization
in $\zm{N-1}{m}={\rm e}^{-\jmath \frac{2\pi m}{N-1}}$ \cite{Guo:2022:J}: 
\begin{equation}
\sum_{n=0}^{N-2} \bar{g}_{i}\sqb{n} \, \zm{N-1}{n\cdot m}=\frac{\PMi ( \zm{N-1}{m} )} {\QMs ( \zm{N-1}{m})},  \quad 
\begin{array}{*{20}{l}}
\PMi \in P_{K-1} \\ 
\QMs \in P_{K}
\end{array}
\label{eq:poly}
\end{equation}
where the denominator $\QMs\rob{z}$ \eqref{eq:multi-channel model} is the annihilation filter \emph{common} to all channels; its roots uniquely determine the frequencies $u_k$. Let $\widehat{\bar{g}}_{i}$ be the DFT of ${\bar{g}}_{i} \sqb{n}$, the \usfmc samples can be expressed as
\begin{equation}
\label{eq:multi-channel model}
\widehat{\bar{g}}_{i} [m] = \frac{\PMi ( \zm{N-1}{m} )} {\QMs ( \zm{N-1}{m})} ,\ i \in\id{4}.
\end{equation}
Parseval’s identity implies \eqref{eq:block matrix} and \eqref{eq:poly} are equivalent. Hence,
\begin{equation}
\label{eq:vector_fitting}
\mathop {\min}\limits_{{\QMs},{\PMi}} \  \sum\limits_{i = 0}^{3} \sum\limits_{m = 0}^{N-2} {{{\left| {{ \widehat{\bar{g}}_{i} [m]  } - \frac{ \PMi ( \zm{N-1}{m} ) }{\QMs ( \zm{N-1}{m})}} \right|}^2}}.
\end{equation}
To solve \eqref{eq:vector_fitting}, we adopt an iterative strategy \ie we construct a collection of estimates for $\QMs$, and select the one that minimizes the mean-square error (MSE) of the \usfmc data via \eqref{eq:vector_fitting}. These estimates $\{ {\mathsf{H}^{[j]}} \}$ are found iteratively by solving the approximate problem (since ${\mathsf{H}^{[j-1]}}\approx {\mathsf{H}}$)
\begin{equation}
\min\limits_{\mathsf{H},\PMi} \sum\limits_{i = 0}^{3} \sum\limits_{m = 0}^{N-2} {{{\left| { \frac{{\widehat{\bar{g}}_{i} [m] }{\mathsf{H}\left( \zm{N-1}{m} \right)} -{{\PMi}\left( \zm{N-1}{m} \right)}}{{{\mathsf{H}^{[j-1]}}\left( \zm{N-1}{m} \right)}}} \right|}^2}}, 
\quad
j\in \id{j_{\max}}.
\label{eq:linear_fitting}
\end{equation}
Initializing ${\mathsf{H}^{[j]}}$ differently provides estimation diversity.

With $\QMs$ and $\PMi$ known, the frequencies $\nu_k$ are obtained by $\mathsf{roots}\rob{\QMs}\mapsto u_k$ since $\mathsf{H}\rob{u_k} = 0 $. The corresponding amplitudes can be calculated via least-squares method,
\begin{equation}
\label{eq:parameter estimates}
\left\{
\begin{aligned}
&\nu_{k}=\tfrac{ 1}{T} \mathsf{Im}\bigl( \log(u_k)\bigr), \\
&c_{k,i} =-{\tfrac{{  u_k {\PMi}\left( u_k^{-1} \right)}}{{\left( {1 - u_k^{N-1}} \right)
{\left. {{\partial_{z}\QMs}\left( z \right)} \right|_{z = {u_k^{-1}}}}}}}.
\end{aligned} 
\right.
\end{equation}
Having estimated the parameters $\{\nu_{k}, c_{k,i}\}_{k\in\id{K}}^{i\in\id{4}}$, the multi-channel sinusoidal samples can be reconstructed using~\eqref{eq:global description}. 

\bpara{Algorithmic Implementation.} We provide an algorithm for \eqref{eq:linear_fitting}. With ${\widehat{\bar{\mathbf{g}}}}_{i} =\mathbf{V}_{N-1}^{N-1}{{\bar{\mathbf{g}}}}_{i}$,
the trigonometric polynomials $$\{\PMi(\zm{N-1}{m}), \QMs(\zm{N-1}{m})\}_{m\in \id{N-1}}$$ in \eqref{eq:linear_fitting} are written as, 
\[
\bigl[\PMi(\zm{N-1}{m})\bigr] = \mathbf{V}_{N-1}^{K} \mathbf{q}_{i}
\ \mbox{ and } \
\bigl[\QMs(\zm{N-1}{m})\bigr] = \mathbf{V}_{N-1}^{K + 1} \mathbf{h}
\]
where $\mathbf{q}_{i}$ and $\mathbf{h}$ are the coefficients of $\PMi \in P_{K-1}$ and $\QMs \in P_{K}$, stacked in vector form and denote $\mat{q} = \vect{\mat{q}_i}{4}$. Assuming the estimate $\mathbf{h}^{[j]}$ of $\mathbf{h}$ at iteration-$j$ is known, the minimization at $j+1$ can be reformulated in matrix form as,  
\begin{equation}
\label{eq:vector frequency matrix form}
\{{\bf q}^{[j+1]},{\bf h}^{[j+1]}\}=\mathop{\rm arg\,min}\nolimits_{{\mathbf{{q}}}, \mathbf{h}}\left\|{\mathbf{A}}^{[j]} \mathbf{h}-{\mathbf{B}}^{[j]} {\mathbf{{q}}}\right\|^{2}_2
\end{equation}
where $\mat{R}^{[j]}, \mat{A}^{[j]}$ and $\mat{B}^{[j]}$ and their decompositions are respectively given by \eqref{eq:matrix} (on the top of this page). 
\begin{figure*}
\[
\small
{\begin{array}{*{20}{c|}}
{\mathbf{A}}^{[j]} = {\mathbf{D}}  {\bf R}^{[j]} \mathbf{V}_{N-1}^{K+1} &
{\mathbf{B}}^{[j]} ={\bf I}_{4}\otimes\rob{{\bf R}^{[j]}\mathbf{V}_{N-1}^{K}} &
{\bf R}^{[j]} = \left(\dk \left( \mathbf{V}_{N-1}^{K + 1} \mathbf{h}^{[j]} \right)\right)^{-1} &
{\mathbf{D}} = \begin{bmatrix}
\dk({\widehat{\bar{\mathbf{g}}}}_{0}),
\dk({\widehat{\bar{\mathbf{g}}}}_{1}),
\dk({\widehat{\bar{\mathbf{g}}}}_{2}),
\dk({\widehat{\bar{\mathbf{g}}}}_{3})
\end{bmatrix}^{\transp}
\end{array}}\]
\begin{equation}
\label{eq:matrix}
\tag{20}
\renewcommand\matscale{.51}
\abmat{8}{4}{4(N-1)}{(K+1)}{A^{\mathrm{[j]}}} = 
\abmat{8}{5}{4(N-1)}{(N-1)}{D} 
\raiserows{1.5}{\abmat{5}{5}{(N-1)}{(N-1)}{{R}^{[j]}}} 
\raiserows{1.5}{\abmat{5}{3}{(N-1)}{(K+1)}{\mathbf{V}}}
\quad 
\mbox{ and }
\quad \ \ 
\abmat{8}{6}{4(N-1)}{4K}{{\mathbf{B}}^{[j]} } = 
\abmat{3}{3}{4}{4}{{\bf I}}\otimes 
\left(
\raiserows{0}{\abmat{5}{5}{(N-1)}{(N-1)}{{\bf R}^{[j]}}} \raiserows{0}{\abmat{5}{3}{(N-1)}{(K)}{\mathbf{V}}}
\right)
\vspace{-5pt}
\end{equation}
\vspace{-5pt}
\rule{\textwidth}{0.8pt}
\end{figure*}
\setcounter{equation}{20}
Notice that $\mathbf{B}^{[j]}$ is full rank because $\mathbf{R}^{[j]}$ and $\mathbf{V}_{N-1}^{K}$ are full rank. We adopt a normalization constraint 
\begin{equation}
\label{eq:normalization constraint}
\frac{1}{2\pi}\int_{0}^{2\pi}\overline{\QMs^{[0]}(e^{-\jmath\theta})}\QMs^{[j+1]}(e^{-\jmath\theta})d\theta=1
\end{equation}
to ensure the uniqueness of the optimal solution to~\eqref{eq:vector frequency matrix form}, where $\QMs^{[0]} $ is the initialization of the algorithm. 
Consequently, the quadratic minimization~\eqref{eq:vector frequency matrix form} can be eventually posed as
\begin{tcolorbox}[ams align] 
& \{{\bf q}^{[j+1]},{\bf h}^{[j+1]}\}=\mathop{\rm arg\,min}\limits_{{\mathbf{{q}}}, \mathbf{h}}\left\|{\mathbf{A}}^{[j]} \mathbf{h}-{\mathbf{B}}^{[j]} {\mathbf{{q}}}\right\|^{2} \\
&\mbox{subject to} \  \frac{1}{2\pi} \int_{0}^{2\pi}\overline{\QMs^{[0]}(e^{-\jmath\theta})}\QMs^{[j+1]}(e^{-\jmath\theta})d\theta=1  \nonumber
\end{tcolorbox}
which results in the update 
\begin{equation}
\begin{bmatrix}{\bf h}^{[j+1]}\\-{\bf q}^{[j+1]}\end{bmatrix}=\kappa\left(\bigl[{\bf A}^{[j]},{\bf B}^{[j]}\bigr]^{\sf H}\bigl[{\bf A}^{[j]},{\bf B}^{[j]}\bigr]\right)^{-1}\begin{bmatrix}{\bf h}^{[0]}\\0\end{bmatrix}
\label{eq:vector frequency solution}
\end{equation}
where $\mathbf{h}^{[0]}$ are the coefficients of $\QMs^{[0]} $ and $\kappa$ is the Lagrange multiplier such that the normalization constraint~\eqref{eq:normalization constraint} is satisfied\footnote{For each initialization $\mathbf{h}^{[0]}$, we update the trigonometric polynomial coefficients $\mathbf{h}^{[j+1]}$ until the stopping criterion~\eqref{eq:stopping criterion} is met. If \eqref{eq:stopping criterion} is not met for certain choices of $\mathbf{h}^{[j+1]}$ after reaching the maximum iteration count $j_{\rm max} $, we restart the algorithm with a different initialization.}. The procedure is summarized in \algref{alg:snu}.

\subsection{Sub-Problem: Robust Residue Separation.}

With $\mathbf{\bar{g}}_{i}$ known from the method in \secref{sec:spone}, the minimization on $\mat{u}_{i}$ essentially boils down to a convex optimization problem.
The main result is as follows.
\begin{theorem}
\label{theorem: residue estimation}
Given signal estimates ${\bar{g}}_{i} =  u_{i}+ v_{i}, \ i \in\id{4}$, the optimal solution $u_i^{\opt}$ to \eqref{eq:global problem} is given by,
\begin{tcolorbox}[ams align, top = 0pt, bottom = 0pt,breakable]
\label{eq:residue estimation problem}
&\mathbf{u}_{i}^{\opt} = 
\begin{cases}
\tr{i} \left(\mathbf{\bar{g}}_{0} + \mathbf{\bar{g}}_{1} - 2 \mathbf{v}_{i} , \left(-1\right)^{i} \left(\mathbf{\bar{g}}_{0} - \mathbf{\bar{g}}_{1}\right)\right), i=0,1  \\
\tr{i} \left(\mathbf{\bar{g}}_{2} + \mathbf{\bar{g}}_{3} - 2 \mathbf{v}_{i} , \left(-1\right)^{i}\left(\mathbf{\bar{g}}_{2} - \mathbf{\bar{g}}_{3}\right)\right), i=2,3
\end{cases}\nonumber \\
&\tr{} \left(\cdot \right)= \left( \mathscr{Q}_\lambda  \circ\th\right) \left(\cdot\right), \ \lambda_i = \lambda_{\rob{i\bmod 2}} \nonumber \\
&\th \left(x, y\right) = \frac{1}{2} \left(x + \operatorname{sgn}\left(y\right) \mathop {\min} \left(  \left| y\right|, \sigma \right)  \right).
\end{tcolorbox}
\end{theorem}
\begin{proof}
Given $\mathbf{\bar{g}}_{i}$ using \eqref{eq:global description} via \eqref{eq:parameter estimates}, \eqref{eq:global problem} amounts to a constrained minimization on $\{\mat{u}_{i}\}_{i\in\id{4}}$. Thus, we reformulate \eqref{eq:global problem},
\begin{align}
\label{eq:residue formulation}
\PT \ \	&\mathop {\min}
_{\substack{
\mathbf{u}_{i}	}} \ \  \sum\nolimits_{i\in\id{4}} \norm{\mathbf{u}_{i} + \mathbf{v}_{i} - \mathbf{\bar{g}}_{i}}_{2}^{2}, \ \mbox{s.t.} \ \mat{u}_{i} \in 2\lambda_{i} \mathbb{Z} \nonumber \\
&  \norm{ \mathbf{u}_{i}+\mathbf{v}_{i} -\mathbf{u}_{i+1}-\mathbf{v}_{i+1} }_{\infty}\leqslant \sigma,  \ i= \{0,2\}.
\end{align}
Consider the above minimization~\eqref{eq:residue formulation} without the on-grid constraints. Let $\vecd{i}{\mat{u}} = \mat{u}_{i+1} - \mat{u}_{i}$ denotes the vector difference. By rewriting the $\infty$-norm constraints, we have,
\begin{align}
\label{eq:residue estimation problem-2}
\mathop {\min}
_{\substack{
u_{i}\sqb{n} 	}} \quad   &\sum\nolimits_{i\in\id{4}} \norm{\mathbf{u}_{i} + \mathbf{v}_{i} - \mathbf{\bar{g}}_{i}}_{2}^{2}  \\
\mbox{subject to} \ \  &  -\sigma  \mathbf{1} \preccurlyeq \vecd{i}{\mat{u}+\mat{v}} \preccurlyeq  \sigma  \mathbf{1}, \ \  i=0,2  \nonumber
\end{align} 
where $\preccurlyeq$ denotes element-wise smaller or equal and $\mathbf{1} = \left[1,\cdots,1\right]^{\transp} \in \mathbb{R}^{N-1}$.  Notice that, the minimization on $u_0, u_1$ is independent of $u_2, u_3$, thus~\eqref{eq:residue estimation problem-2} can be split into two independent sub-problems. Considering the variable symmetry, let's consider the minimization on $u_0, u_1$, leading to 
\begin{align}
\label{eq:residue estimation problem-3}
\mathop {\min}
_{\substack{
\mathbf{u}_{0}, \mathbf{u}_{1}	}} \quad   & \norm{\mathbf{u}_{0} + \mathbf{v}_{0} - \mathbf{\bar{g}}_{0}}_{2}^{2} + \norm{\mathbf{u}_{1} + \mathbf{v}_{1} - \mathbf{\bar{g}}_{1}}_{2}^{2} \\
\mbox{subject to} \ \  & -\sigma  \mathbf{1} \preccurlyeq  \vecd{0}{\mat{u}+\mat{v}} \preccurlyeq  \sigma  \mathbf{1}. \label{eq:constraints-1}
\end{align}
We follow Lagrange-multiplier approach for the minimization. To this end, we define, 
\[
\mathcal{L} \left( 	\mathbf{u}_{0}, \mathbf{u}_{1}\right) = \sum_{i=0}^{1}\norm{\mathbf{u}_{i} + \mathbf{v}_{i} - \mathbf{\bar{g}}_{i}}_{2}^{2} +{\bm \beta}_{0}^{\transp} (-\vecd{0}{\mat{u}+\mat{v}} -  \sigma  \mathbf{1}) + {\bm \beta}_{1}^{\transp} ( \vecd{0}{\mat{u}+\mat{v}} -  \sigma  \mathbf{1}),
\]
and $\vecd{i}{\mat{u}} = \mat{u}_{i+1} - \mat{u}_{i}$ denotes the vector difference. The minimizer satisfies  
\begin{align*}
{\bm \beta}_{0}^{\transp} (-\vecd{0}{\mat{u}+\mat{v}} -  \sigma  \mathbf{1})  = \mathbf{0}, \quad {\bm \beta}_{1}^{\transp} ( \vecd{0}{\mat{u}+\mat{v}} -  \sigma  \mathbf{1}) = \mathbf{0}.
\end{align*}
Taking the derivative results in 
\[
\left\{ {\begin{array}{*{20}{l}}
\frac{\partial \mathcal{L}  \left( 	\mathbf{u}_{0}, \mathbf{u}_{1}\right)}{\partial u_{0}\sqb{n}} = 
\rob{u_{0} + v_{0} - {\bar{g}}_{0}  + {\bm \beta}_{0} + {\bm \beta}_{1}} \sqb{n}= 0 \\
\frac{\partial \mathcal{L}  \left( 	\mathbf{u}_{0}, \mathbf{u}_{1}\right)}{\partial u_{1}\sqb{n}} =  
\rob{u_{1} + v_{1} - {\bar{g}}_{1} - {\bm \beta}_{0} - {\bm \beta}_{1}}\sqb{n}= 0
\end{array}} \right. .
\]
Notice that, the LHS and RHS of \eqref{eq:constraints-1} cannot hold simultaneously due to its conflicting nature. For instance, ${\bm \beta}_{1}[n]=0$ if the LHS of \eqref{eq:constraints-1} is active. To simplify expressions, let $\eta_{i}\sqb{n}= ({\bar{g}_{0} + \bar{g}_{1} - 2v_{i}})\sqb{n}$ and $\zeta\sqb{n} = ({\bar{g}}_{0} - {\bar{g}}_{1})\sqb{n}$. Thus, the optimal solution can be categorized into three cases:
\begin{enumerate}[leftmargin=*, label = \fbox{\ \#\arabic*\ },labelsep = 12pt] 
\itemsep -4pt
\item $| \zeta\sqb{n} | < \sigma$. In this situation, both LHS and RHS of \eqref{eq:constraints-1} are loose or inactive, resulting in
\begin{equation}
\label{eq:0 sigma}
u_{i}^{\opt}\sqb{n} =\tfrac{1}{2} (\eta_{i}\sqb{n} + (-1)^{i}  \zeta\sqb{n}), \ i=\{0,1\}.
\end{equation}
\item $\zeta\sqb{n} \leqslant -\sigma$. In this scenario, the RHS of \eqref{eq:constraints-1} is tight, 
\begin{equation}
\label{eq:-sigma}
u_{i}^{\opt}\sqb{n} =  \tfrac{1}{2} ( \eta_{i}\sqb{n} - (-1)^{i}   \sigma), \ i=\{0,1\}.
\end{equation}
\item $ \zeta\sqb{n} \geqslant \sigma$. In this scenario, the LHS of \eqref{eq:constraints-1} is tight, 
\begin{equation}
\label{eq:+sigma}
u_{i}^{\opt}\sqb{n} =  \tfrac{1}{2} ( \eta_{i}\sqb{n} + (-1)^{i}   \sigma), \ i=\{0,1\}.
\end{equation}
\end{enumerate} 
From \eqref{eq:0 sigma}-\eqref{eq:+sigma}, we can characterize the optimizer as ($i=0,1$)
\begin{equation}
\label{eq:optimal solution formula}
u_{i}^{\opt}\sqb{n} =  \frac{1}{2} ( \eta_{i}\sqb{n} + (-1)^{i} \sgn (\zeta\sqb{n})  \mathop {\min} (  \left| \zeta\sqb{n}\right|, \sigma ))
\end{equation}
which leads to a thresholding operator in \eqref{eq:residue estimation problem}. Similarly, the same calculation applies to $u_{i}^{\opt}, i=2,3$. Enforcing $u_{i}^{\opt}$ on the grid $2\mathbb{Z}\lambda_{i}$ can be implemented via the quantization operator $\mathscr{Q}_{\lambda_i}
=2\lambda_i \flr{\rob{ \cdot +\lambda_i}/\rob{2\lambda_i}}$,
which completes the proof.
\end{proof}

\begin{algorithm}[!t]
\setstretch{0.9}
\caption{\abox{{\bf Algorithm 1.} \snu: \snyq \usf Spectral Estimation}}
\label{alg:snu}
\begin{algorithmic}[1]
\REQUIRE Multi-channel samples $\{{\bf 	{\bar{g}}}_{i}\}_{i\in\id{4}}$
\STATE Compute the DFT of 
${\bf{\bar{g}}}_{i}$: ${\widehat{\bar{\mathbf{g}}}}_{i} =\mathbf{V}_{N-1}^{N-1}{{\bar{\mathbf{g}}}}_{i}, i\in\id{4}$
\FOR{\texttt{loop = 1} to \texttt{max. initializations}}
\STATE Initialize $\mathbf{h}$ as $\mathbf{h}^{[0]}$;
\FOR{$j=1$ to $j_{\rm max}$}
\STATE Construct the matrices in~\eqref{eq:matrix} \\
\STATE Update $\mathbf{h}^i$ and ${\bf q}^{i}$ by solving~\eqref{eq:vector frequency solution};
\IF{\eqref{eq:stopping criterion} holds}
\STATE \texttt{Terminate all loops};
\ENDIF
\ENDFOR
\ENDFOR
\STATE $\mathbf{h}=\mathbf{h}^{[j]}$, ${\mathbf{{q}}}={\mathbf{{q}}}^{[j]}$;\\
Calculate $\{\nu_{k}, c_{k,i}\}_{k\in\id{K}}^{i\in\id{4}}$ using~\eqref{eq:parameter estimates}. 
\ENSURE The sinusoidal parameters $\{\nu_{k}, c_{k,i}\}_{k\in \id{K}}^{i\in\id{4}}$. 
\end{algorithmic}
\end{algorithm}

\bpara{Stopping Criterion.} We initialize the \rsnu method in \algref{alg:rsnu} by computing $\{\mathbf{u}_{i}^{[0]}\}_{i\in\id{4}}$ in \eqref{eq:residue difference}, which is a reasonable intialization. Using \eqref{eq:residue estimation problem}, we then estimate $\{\nu_{k}, c_{k,i}\}_{k\in\id{K}}^{i\in\id{4}}$ via \eqref{eq:parameter estimates} based on which we refine $\{\mathbf{u}_{i}^{[0]}\}_{i\in\id{4}}$ via \eqref{eq:residue estimation problem}. The following 
\begin{equation}
\label{eq:stopping criterion}
\mathop {\max} \left( \norm{\mathbf{\bar{g}}_{0} - \mathbf{\bar{g}}_{1}  }_{\infty}, \norm{ \mathbf{\bar{g}}_{2} - \mathbf{\bar{g}}_{3} }_{\infty} \right)  \leqslant \sigma
\end{equation}
is used as our stopping criterion. Iterating the method leads to robust estimates as it eliminates distortion, which is validated via hardware experiments. 

\bpara{Algorithmic Complexity.}
\rg{In terms of computational complexity, the proposed \rsnu method requires residue initializations via look-up table and QR decompositions of matrices ${\bf A}^{[j]}\in\C^{4(N-1)\times(K+1)}$ and ${\bf B}^{[j]}\in\C^{ 4(N-1)\times 4K }$ in joint spectral estimation (most time-consuming). The QR decomposition of ${\bf A}^{[j]}$ essentially amounts to performing $K+1$-times {Gram–Schmidt orthogonalization} as $K<N$. Moreover, notice that, ${\bf B}^{[j]}$ is sparse due to the Kronecker operation described in \eqref{eq:matrix}. These algebraic structures allow for effective acceleration in matrix computations. Hence, the total computational complexity mainly scales with the number of sinusoids $K$ and varies slowly with the number of samples $N$, which significantly increases the algorithm run-speed when processing large-scale hardware data (see the last column of \tabref{tab:exp1} and \tabref{tab:exp2}). 
}

\begin{algorithm}[!t]
\setstretch{0.9}
\caption{\abox{{\bf Algorithm 2.}
\rsnu: Robust Version of \snu Method}}
\label{alg:rsnu}
\begin{algorithmic}[1]
\REQUIRE Folded Samples $\{\mathbf{y}_{i}\}_{i\in\id{4}}$.
\STATE Compute the finite difference $\{\mathbf{v}_{i}\}_{i\in\id{4}}$.
\STATE Compute the initial estimates $\{\mathbf{u}_{i}^{[0]}\}_{i\in\id{4}}$ by using~\eqref{eq:residue difference}.
\FOR{\texttt{loop = 1} to \texttt{max. iterations}}
\STATE Update $\{\nu_{k}, c_{k,i}\}_{k\in\id{K}}^{i\in\id{4}}$ using \algref{alg:snu}; 
\STATE Update $\{\mathbf{u}_{i}\}_{i\in\id{4}}$ using~\eqref{eq:residue estimation problem};
\IF{\eqref{eq:stopping criterion} holds}
\STATE \texttt{Terminate all loops};
\ENDIF
\ENDFOR
\STATE Calculate $\{\omega_{k}, c_{k,i}\}_{k\in\id{K}}^{i\in\id{4}}$ using~\eqref{eq:global description};
\ENSURE The sinusoidal parameters $\{\omega_{k}, c_{k,i}\}_{k\in\id{K}}^{i\in\id{4}}$. 
\end{algorithmic}
\end{algorithm}

\begin{figure}[!t]
\begin{center}

\includegraphics[width =0.45\textwidth]{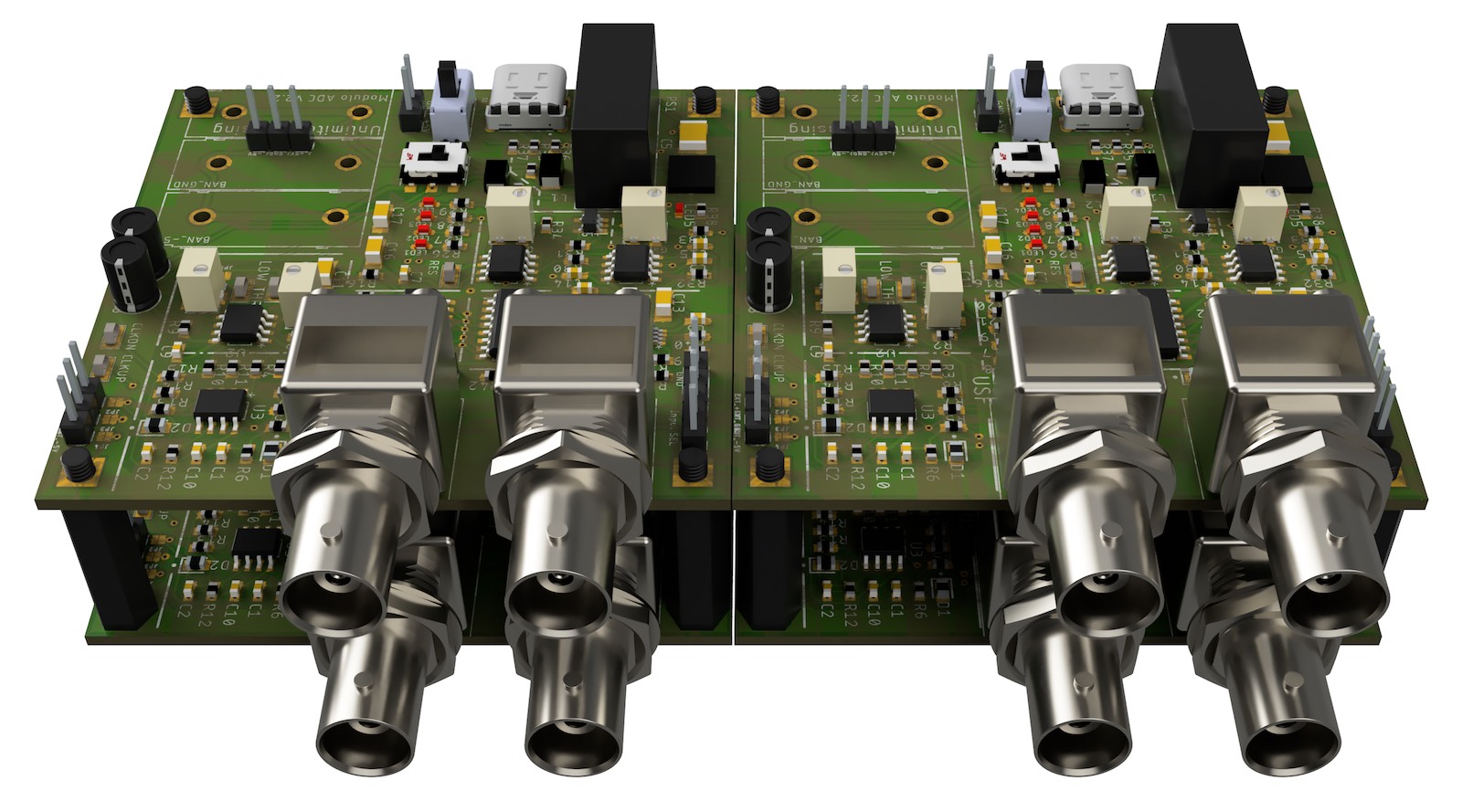}
\end{center}
\caption{Custom-designed \madc hardware prototype for robust \snyq spectrum estimation (\rsnu) experiments reported in \secref{sec:experiments}. Our testbed is designed such that experimental parameters, such as the sampling step $T$, time delay $T_d$ and thresholds $\lambda_{0},\lambda_{1}$ are tunable.}
\label{fig:proto}
\end{figure}

\begin{figure}[!t]
\begin{center}

\includegraphics[width =0.5\textwidth]{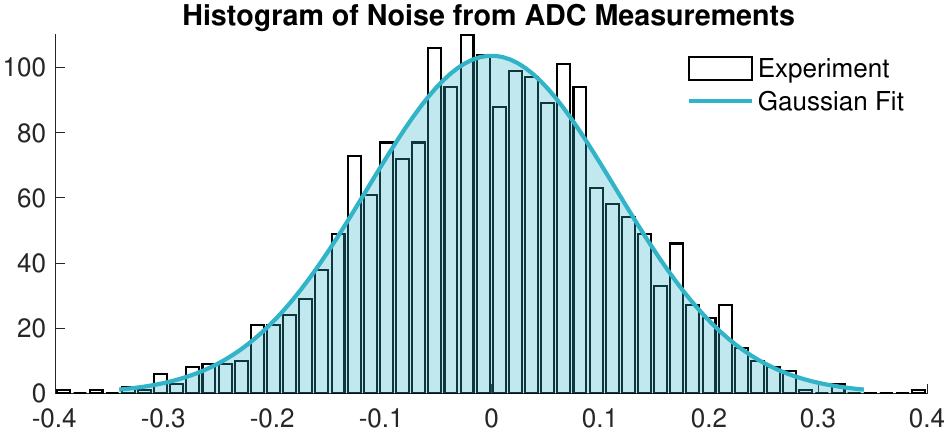}
\end{center}
\caption{{\color{black}Noise characterization via hardware experiments. We plot the noise histogram by estimating noise at the output of \madc hardware \cite{Bhandari:2021:J}. Gaussian approximation of the histogram provides a proxy for noise statistics with approximately zero mean and standard deviation $=0.11$, translating to about $19.40$ dB signal-to-noise ratio (SNR).}}
\label{fig:noise}
\end{figure}

\section{Experiments}
\label{sec:experiments}

The overarching goal of our experiments is to computationally push the frequency range of our hardware by using \snyq spectral estimation methods. In particular, through a series of $14$ experiments, we show that HDR signals in the \emph{kilohertz} range can be estimated via \snyq sampling at \emph{hertz} range, utilizing the \madc hardware, providing a factor of $1000\times$ improvement in real-world scenarios. This also serves as a validation of our \rsnu method (Algorithm~\ref{alg:rsnu}) developed in \secref{sec:robust reconstruction}. 

In noiseless scenarios, numerical simulations demonstrate that the method presented in \secref{sec:methodology} results in an exact signal recovery (up to the machine precision). Hence, we report only hardware experiments, \rgn{which are more challenging. As discussed in \secref{sec:robust reconstruction}, one has to overcome data distortion arising from system noise (including thermal and quantization noise following Gaussian and Uniform distributions, respectively) and from folding non-idealities \cite{Bhandari:2021:J}, to achieve successful recovery.} 

\rgn{In our case, as shown in \fig{fig:noise}, experimental observation of the noise from \madc follows  Gaussian statistics, approximately. Though our goal is not to validate the precise nature of the distribution, we aim to quantify the system’s signal-to-noise ratio (SNR). Given that the measurement noise can be explained by an approximately zero mean and a standard deviation of $0.11$, this translates to an SNR of $19.4$ dB.}

\rgn{Our method outlined in \algref{alg:rsnu} is designed to tackle noisy scenarios. We do so by adapting  $\sigma$ in \eqref{eq:global problem} that uses a proxy of the SNR as an input to the algorithmic framework. Empirically, we have observed that $\sigma = 2\alpha\max \left|\lambda_{i}\right|/(2^B-1)$ is an effective estimate, where $B$ is the bit budget and $\alpha>0$ is a constant depending on noise or distortion. For low-frequency inputs ($\leqslant 1$ \khz), we set $\alpha = 2$ as $\sigma$ is dictated by quantization and system noise. As for high-frequency inputs ($\geqslant 1$ \khz), we set $\alpha=4$ as non-ideal foldings are dominant, inducing \textit{significant data distortion challenges}.}

\bpara{Multi-Channel Hardware.} Our multi-channel sampling hardware is shown in \fig{fig:proto} and is based on the \madc introduced in \cite{Bhandari:2021:J}. We particularly, use off-the-shelf electronic components resulting in an inexpensive design implementing the \usfmc architecture in \fig{fig:architecture}.

\begin{table*}[!t]
\centering
\caption{Hardware Based Experimental Parameters and Performance Evaluation for Low-Frequency Inputs.}
\label{tab:exp1}
\resizebox{\textwidth}{!}{
\begin{tabular}{@{}lccccccccccc@{}}
\toprule
\multicolumn{1}{c}{Figure}  & $N$ & $\fs$ & $T_d$ & $\lambda_0$ & $\lambda_1$ & $\norm{g}_{\infty}$ & $f_k$ &$\widetilde f_k$ &  $\msef$ & Run-Time
\\ \midrule
&                              &               (\hz)            &  ($\mu$s)    & (V) & (V) & (V)               & (\khz)& (\khz) & & (s) \\ \midrule
\multicolumn{1}{c}{\fig{fig: hardware experiment-1} (a)}&$200$&$877$&$200$&$0.98$&$1.88$&$8.88$&$[0.4,0.7,1.0]$&$[0.400,0.700,1.000]$&$2.03\times10^{-6}$&\rg{$6.10\times10^{-1}$}\\
\multicolumn{1}{c}{\textemdash}&$200$&$461$&$200$&$0.98$&$1.88$&$8.92$&$[0.4,0.7,1.0]$&$[0.400,0.700,1.000]$&$6.32\times10^{-7}$&\rg{$5.73\times10^{-1}$}\\
\multicolumn{1}{c}{\textemdash}&$200$&$211$&$200$&$0.98$&$1.88$&$8.88$&$[0.4,0.7,1.0]$&$[0.400,0.700,1.000]$&$1.78\times10^{-6}$&\rg{$5.70\times10^{-1}$}\\
\multicolumn{1}{c}{\textemdash}&$100$&$89$&$200$&$0.98$&$1.88$&$8.90$&$[0.4,0.7,1.0]$&$[0.400,0.700,1.000]$&$1.16\times10^{-6}$&\rg{$3.94\times10^{-1}$}\\
\multicolumn{1}{c}{\textemdash}&$100$&$79$&$200$&$0.98$&$1.88$&$8.99$&$[0.4,0.7,1.0]$&$[0.400,0.700,1.000]$&$1.71\times10^{-7}$&\rg{$3.54\times10^{-1}$}\\
\multicolumn{1}{c}{\fig{fig: hardware experiment-1} (b)}&$100$&$59$&$200$&$0.98$&$1.88$&$8.97$&$[0.4,0.7,1.0]$&$[0.400,0.700,1.000]$&$5.08\times10^{-7}$&\rg{$3.86\times10^{-1}$}\\
\multicolumn{1}{c}{\textemdash}&$100$&$41$&$200$&$0.98$&$1.88$&$9.01$&$[0.4,0.7,1.0]$&$[0.400,0.700,1.000]$&$3.98\times10^{-7}$&\rg{$3.80\times10^{-1}$}\\
\multicolumn{1}{c}{\fig{fig: hardware experiment-1} (c)}&$100$&$29$&$200$&$0.98$&$1.88$&$8.82$&$[0.4,0.7,1.0]$&$[0.400,0.700,1.000]$&$4.33\times10^{-7}$&\rg{$4.08\times10^{-1}$}\\
\bottomrule
\end{tabular}
}
\end{table*}

\begin{figure*}[!t]
\centering
\includegraphics[width=\linewidth]{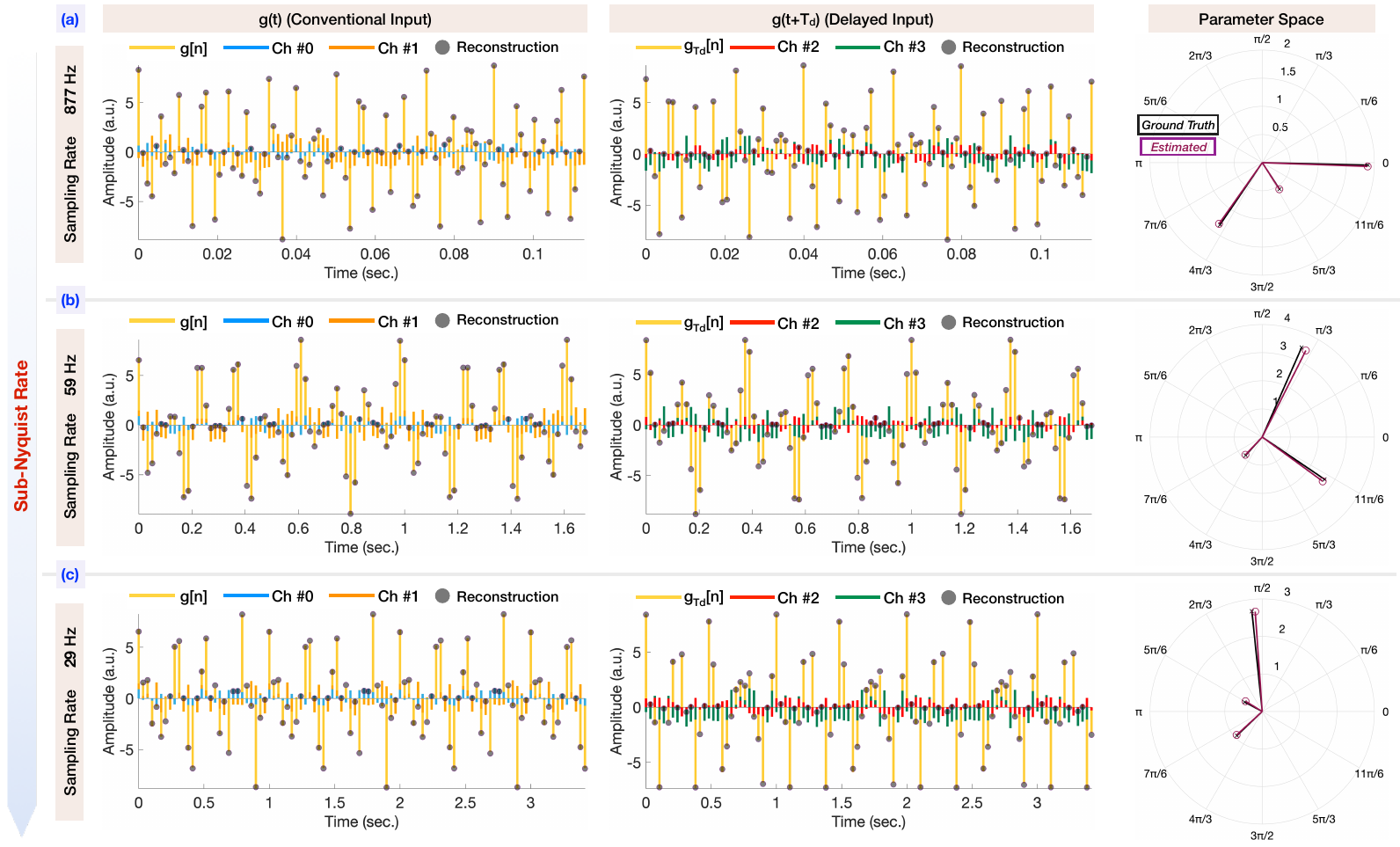}
\caption{Hardware validation with low-frequency inputs. The ground-truth frequencies are $\{f_k\}_{k\in\id{K}}=[400,700,1000]$ \hz. We conduct 3 experiments with decreasing sampling rates: (a) 877 \hz, (b) 59 \hz and (c) 29 \hz. The signal recoveries are shown in the first two columns; the estimated parameters $\{\widetilde c_k,\widetilde \omega_k\}_{k\in \id{K}}$ are plotted in phasor domain in the last column. The HDR signals are successfully recovered in all realizations with MSE bounded by $10^{-6}$.}
\label{fig: hardware experiment-1}
\end{figure*}

\bpara{Experimental Protocol.} 
We use \texttt{TG\-5011A} as our HDR analog signal generator. Its output is fed to the \usfmc hardware with tunable thresholds ($\lambda_i$). Together with the \usfmc output, we simultaneously record the original HDR input on the \texttt{PicoScope 3406D} oscilloscope, serving as the ground-truth $\{\gn, \gdn\}$.
We sample the \usfmc output with $6$-bit resolution, particularly highlighting the algorithmic capability in the presence of quantization noise. That said, the ground-truth $\{\gn, \gdn\}$ is sampled with $7$-bits to accommodate the HDR swing of the input signal. Experimental parameters such as, ground-truth frequencies $f_k$, dynamic range, sampling frequency $f_s$, delay $T_d$, among others are tabulated in the first row of \tabref{tab:exp1} and \tabref{tab:exp2}, respectively. We conduct $14$ experiments, pushing the limits of current \usfmc hardware. We cover a range of $0.4$ to $7$ \khz signals with \adcDR $\sim 9\lambda$. We study the performance of the \rsnu by considering two scenarios, \viz (i) low-frequency signals ($\leq1$ \khz bandwidth) and (ii) high-frequency signals (signals in range of 4 \khz to 7 \khz). In either case, 
\begin{enumerate}[leftmargin = *, label = \uline{\arabic*})]
\item We consider $K=6$ complex-exponentials mapping to $K'=3$ real sinusoids.
\item Keeping the input fixed, we vary $f_s$ to test the limits of our \rsnu algorithm in a progressively challenging fashion; $f_s$ goes from $877$ \hz to $11$ \hz, invoking a truly \snyq flavor.
\end{enumerate}

\subsection{Experimental Results}

\npara{Low-Frequency Tests.} The input signal frequencies are $\{\mat{f}_k\}_{k\in\id{K}}=[400, 700, 1000]$ \hz and DR is $\norm{g}_{\infty}=9.18\lambda_{0}=4.79\lambda_{1}$. Given this frequency range $\fmax = 1000$ \hz, we set $\alpha = 2$  \footnote{$\alpha\approx 1$ indicates that the quantization dominates the measurement distortion and the non-ideal folding \cite{Bhandari:2021:J} due to hardware imperfection is mild.}. 
\begin{table*}[!htb]
\centering
\caption{Hardware Based Experimental Parameters and Performance Evaluation for High-Frequency Inputs. }
\label{tab:exp2}
\resizebox{\textwidth}{!}{
\begin{tabular}{@{}lccccccccccc@{}}
\toprule
\multicolumn{1}{c}{Figure}  & $N$ & $\fs$ & $T_d$ & $\lambda_0$ & $\lambda_1$ & $\norm{g}_{\infty}$ & $f_k$ &$\widetilde f_k$ &  $\msef$ & Run-Time
\\ \midrule
&                              &               (\hz)            &  ($\mu$s)   & (V) & (V) & (V)               & (\khz)& (\khz) & & (s) \\ \midrule
\multicolumn{1}{c}{\fig{fig: hardware experiment-3} (a)}&$200$&$877$&$50$&$0.98$&$1.88$&$8.73$&$[4,5,6]$&$[4.000,5.000,6.000]$&$2.63\times10^{-4}$&\rg{$5.47\times10^{-1}$}\\
\multicolumn{1}{c}{\textemdash}&$100$&$179$&$50$&$1.30$&$1.46$&$4.76$&$[5,6,7]$&$[4.998,5.998,6.997]$&$5.67\times10^{0}$&\rg{$3.98\times10^{-1}$}\\
\multicolumn{1}{c}{\textemdash}&$100$&$79$&$50$&$1.30$&$1.46$&$4.68$&$[5,6,7]$&$[4.998,5.998,6.997]$&$5.81\times10^{0}$&\rg{$3.81\times10^{-1}$}\\
\multicolumn{1}{c}{\textemdash}&$100$&$29$&$50$&$1.30$&$1.46$&$3.97$&$[5,6,7]$&$[4.999,5.999,6.999]$&$1.51\times10^{0}$&\rg{$3.71\times10^{-1}$}\\
\multicolumn{1}{c}{\fig{fig: hardware experiment-3} (b)}&$100$&$17$&$50$&$1.20$&$1.45$&$4.05$&$[5,6,7]$&$[4.999,5.999,6.999]$&$1.47\times10^{0}$&\rg{$3.86\times10^{-1}$}\\
\multicolumn{1}{c}{\fig{fig: hardware experiment-3} (c)}&$100$&$11$&$50$&$1.30$&$1.46$&$3.93$&$[5,6,7]$&$[4.999,5.999,6.999]$&$1.47\times10^{0}$&\rg{$4.12\times10^{-1}$}\\
\bottomrule
\end{tabular}
}
\end{table*}
\begin{figure*}[!t]
\centering
\includegraphics[width=\linewidth]{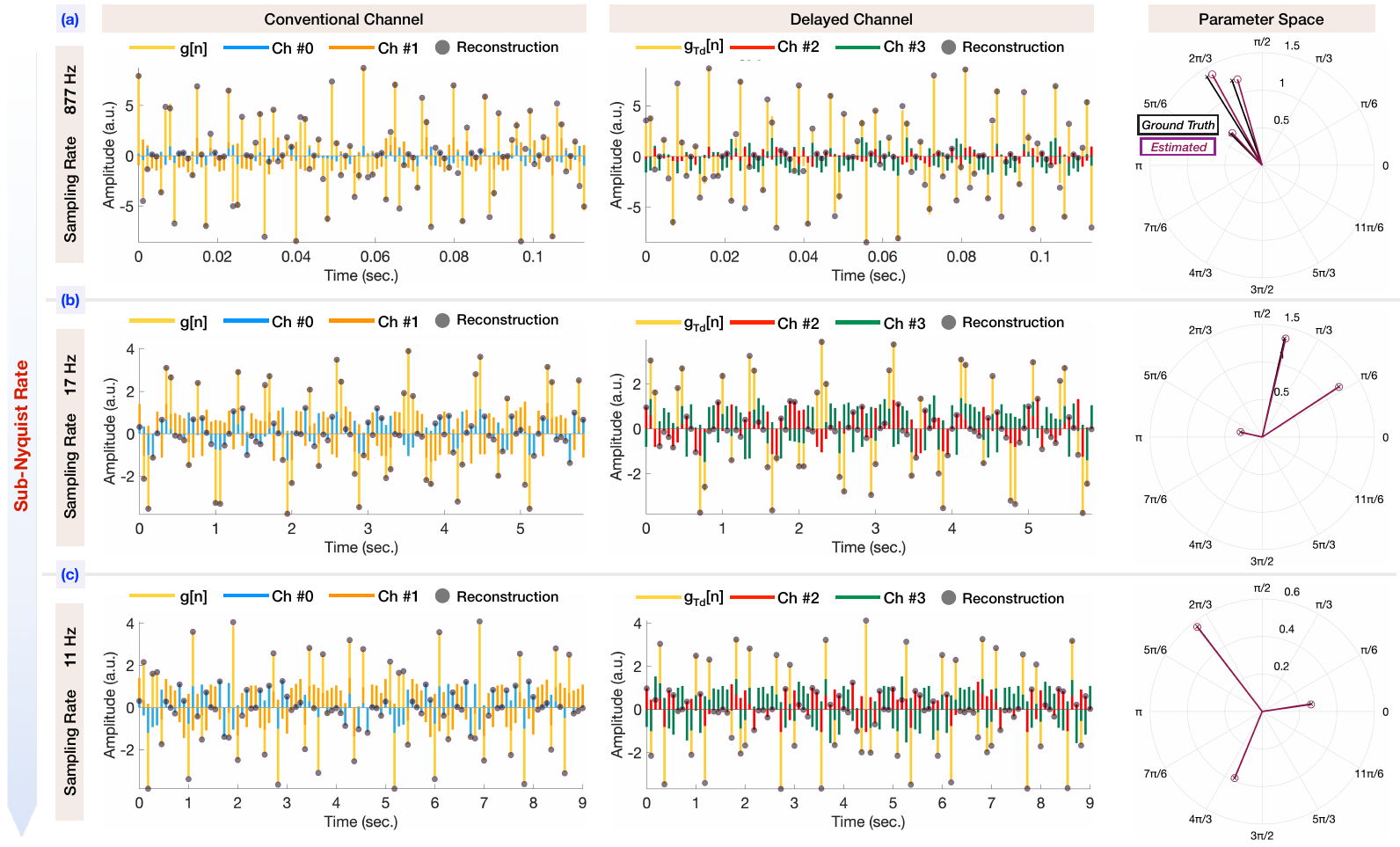}
\caption{Hardware validation with high-frequency inputs. The frequencies are $\{\mat{f}_k\}_{k\in\id{K}}=[4, 5, 6, 7]$ \khz. We conduct 3 experiments with decreasing sampling rates: (a) 877 \hz, (b) 17 \hz and (c) 11 \hz. The signal recoveries are shown in the first two columns; the estimated parameters $\{\widetilde c_k,\widetilde \omega_k\}_{k\in \id{K}}$ are plotted in phasor domain in the last column. The HDR signals are accurately recovered with a low sampling rate (up to $0.078$\%\xspace of the Nyquist frequency), providing a estimation accuracy $ \max_k | {\mat{f}_k - {\widetilde{\mat{f}}}_k} | < 1$ \hz from $6$-bit data, while the conventional method results in $10$-$33$ \hz, even using $7$-bit resolution. 
}
\label{fig: hardware experiment-3}
\end{figure*}
The \rsnu successfully recovers the input signal
$\widetilde{\bar{g}}_{i} = u^{\opt}_{i} + v_{i}$ shown in \fig{fig: hardware experiment-1}, with experiments performed at different \snyq rates. The \usfmc pipeline achieves an accurate spectral estimation with a low sampling rate up to $1.45\%$ of the Nyquist frequency ($2000$ \hz). Together with $9.18\times$ DR improvement, this significantly reduces the sampling cost and power consumption. As shown in \fig{fig: hardware experiment-1}, despite the quantization noise, the HDR signals are accurately reconstructed in all realizations with MSE upper-bounded by $10^{-6}$. The results are summarized in 
\tabref{tab:exp1}.

\npara{High-Frequency Tests.} Here, we further push the frequency range of the \usfmc setup in $4$ to $7$ \khz range, $\{\mat{f}_k\}_{k\in\id{K}}=[4, 5, 6, 7]$ \khz. As before, use   $\sigma = 2\alpha\max \left|\lambda_{i}\right|/(2^B-1)$ with $\alpha=4$ in \eqref{eq:global problem}\footnote{$\alpha \geqslant 3$ indicates that the non-ideal folding due to hardware imperfection dominates the measurement distortion \cite{Bhandari:2021:J}.}. Notice that, \sn with a lower sampling rate is prone to noise interference \cite{Xiao:2018:J}. We use $\msem/f_s = \max | {\mat{f}_k - {\widetilde{\mat{f}}}_k} |/f_s$ to quantify the estimation sensitivity relative to the sampling frequency $f_s$. In hardware experiments, the spectral estimation using the conventional \adc $\{\gn, \gdn\}$ has a large error ($7$-bit resolution): $\msem/f_s=194.12\%$ ($f_s = 17$ \hz), $\msem/f_s=90.91\%$ ($f_s = 11$ \hz),
\begin{enumerate}
\item $f_s = 17$ \hz: estimation is $\left[4.999,6.016,7.033\right]$ \khz. 
\item $f_s = 11$ \hz: estimation is $\left[4.999,6.010,7.010\right]$ \khz. 
\end{enumerate}
where the ground-truth is $\left[5,6,7\right]$ \khz. Despite the non-ideal folding, the robust \snu method or \rsnu attains signal recovery 
$\widetilde{\bar{g}}_{i} = u^{\opt}_{i} + v_{i}$ shown in \fig{fig: hardware experiment-3}, with experiments performed at different \snyq rates from $877$ \hz to $11$ \hz. The \usfmc pipeline achieves spectral estimation with a much lower sampling rate, up to $0.078$\% Nyquist rate. In these settings, non-ideal folding \cite{Bhandari:2021:J} may create challenges for signal recovery. As shown in \fig{fig: hardware experiment-3}, the \khz range frequencies are accurately estimated with $\msem/f_s < 5.88\%$ using $6$-bit resolution, providing $33\times$ accuracy improvement compared to the $7$-bit conventional ADCs. This demonstrates the significance and benefits of \usfmc pipeline that provides more precise spectral estimation induced by higher digital resolution. Together with $8.90\times$ DR improvement, this demonstrates the robustness of the proposed \rsnu algorithm. The results are summarized in \tabref{tab:exp2}.

\section{Conclusion}
Undersampled or sub-Nyquist spectral estimation, aimed at capturing high bandwidth signals where Nyquist rate sampling is impractical or expensive, has been a focus of research for decades. However, most advancements have centered on algorithms and can not handle high-dynamic-range (HDR) signal acquisition. To overcome this fundamental mismatch between theory and practice, our work introduces a multi-channel scheme for sub-Nyquist spectral estimation based on the Unlimited Sensing Framework. Our algorithm utilizes modulo samples, and our recovery theorem guarantees that $K$ frequencies can be estimated from just $6K+4$ measurements, independent of the sampling rate. Hardware experiments with modulo ADCs, coupled with a novel, robust recovery algorithm, indeed show that kilohertz range signals can be recovered at hertz scale sampling rates, providing a true sense of sub-Nyquist spectral estimation. These results serve as a compelling stepping-stone, catalyzing a whole new range of spectral estimation applications, previously unexplored via the lens of USF. Going forward, tightening recovery guarantees, reducing the number of channels, developing robustness guarantees and integrating our approach with new hardware and applications, all remain interesting future research directions.

\bibliographystyle{IEEEtran_URL}


\end{document}